\documentclass[a4paper,UKenglish,cleveref, autoref, thm-restate]{lipics-v2021}
\hideLIPIcs  %uncomment to remove references to LIPIcs series (logo, DOI, ...), e.g. when preparing a pre-final version to be uploaded to arXiv or another public repository

\title{Disjunctions of Two Dependence Atoms} %TODO Please add

\author{Nicolas Fröhlich}{Leibniz Universität Hannover, Germany}{nicolas.froehlich@thi.uni-hannover.de}{https://orcid.org/0009-0003-5413-1823}{The author appreciates funding by the German Research Agency (DFG) under the grant ME2479/3-1 and project id 511769688}

\author{Phokion G.\ Kolaitis}{ University of California Santa Cruz, USA}{kolaitis@ucsc.edu}{https://orcid.org/0000-0002-8407-8563}{}

\author{Arne Meier}{Leibniz Universität Hannover, Germany}{meier@thi.uni-hannover.de}{https://orcid.org/0000-0002-8061-5376}{The author appreciates funding by the German Research Agency (DFG) under the grant ME2479/3-1 and project id 511769688}

\authorrunning{Nicolas Fröhlich, Phokion G.\ Kolaitis, Arne Meier} %TODO mandatory. Please use full first names. LIPIcs license does not support 'et al.' in author lists.

\ccsdesc[500]{Theory of computation~Problems, reductions and completeness}
\ccsdesc[500]{Theory of computation~Logic}
\ccsdesc[500]{Theory of computation~Logic and databases}

\keywords{Dependence logic, coherence, model checking, complexity, functional dependencies} %TODO mandatory; please add comma-separated list of keywords

\nolinenumbers %uncomment to disable line numbering

\usepackage{hyperref}

\usepackage{todonotes}

\usepackage{graphicx}
\usepackage[table]{xcolor}
\usepackage{amsmath,amssymb,amsthm,mathtools}
\usepackage{xspace}
\usepackage{booktabs}
\usepackage[ruled,noline,linesnumbered,noend]{algorithm2e}

\usepackage{tikz}
\usetikzlibrary{
    arrows,
    arrows.meta,
}

\newcommand{\struct}{\mathcal{M}}
\newcommand{\calA}{\struct}

\newcommand{\free}{\operatorname{free}}
\newcommand{\Fr}{\free}
\newcommand{\Dom}{\operatorname{Dom}}

\newcommand{\Range}{\operatorname{Rng}}
\newcommand{\rngrestr}[2]{\Range_{#2}(#1)}
\newcommand{\dep}{\operatorname{dep}}

\newcommand{\Var}{\operatorname{Var}}
\newcommand{\VAR}{\Var}

\newcommand{\Rel}{\operatorname{Rel}}

\newcommand{\FO}{\ensuremath\mathsf{FO}}
\newcommand{\ACzero}{\ensuremath\mathrm{AC^0}}
\newcommand{\Ptime}{\ensuremath\mathrm{P}}
\newcommand{\NP}{\ensuremath\mathrm{NP}}
\newcommand{\NL}{\ensuremath\mathrm{NL}}
\newcommand{\LOGSPACE}{\ensuremath\mathrm{L}}
\newcommand{\DL}{\ensuremath\mathcal{D}}

\newcommand{\ESO}{\ensuremath\mathsf{ESO}}

\newcommand{\MC}{\ensuremath\mathrm{MC}}

\newcommand{\course}{\mathrm{course}}
\newcommand{\dpt}{\mathrm{dpt}}
\newcommand{\inst}{\mathrm{inst}}
\newcommand{\leqlogm}{\leq^{\mathrm{log}}_m}

\newcommand{\STRUC}{\mathrm{STRUC}}
\newcommand{\TwoSAT}{\mathrm{2SAT}}
\newcommand{\specialTwoSat}{\mathrm{mtdf\text-\TwoSAT}}

\newcommand{\problemdef}[3]{%
\begin{center}
\begin{tabular}{r@{\hspace{1ex}}p{11.5cm}}\toprule
\textsf{\bfseries Problem:}& #1 \\\midrule
\textsf{\bfseries Input:}& #2.\\
\textsf{\bfseries Question:}& #3?\\\bottomrule
\end{tabular}
\end{center}
}

\newcommand{\itemstyle}[1]{{\color{lipicsGray}\normalfont\bfseries\sffamily{#1}}}

\usepackage{arydshln}
\newcommand{\rpointer}[2]{{\small(#2\ref{#1})}}

\begin{document}

\maketitle
\begin{abstract}
  Dependence logic is a formalism that augments the syntax of first-order logic with dependence atoms asserting that the value of a  variable is determined by the values of some other variables, i.e., dependence atoms express functional dependencies in relational databases.  
  On finite structures, dependence logic captures $\NP$, hence there are sentences of dependence logic  whose model checking problem is $\NP$-complete. 
  In fact, it is known that there are disjunctions of three dependence atoms whose model checking problem is $\NP$-complete. 
  Motivated from considerations in database theory, we study the model checking problem for disjunctions of two unary dependence atoms and establish a trichotomy theorem, namely, for every such formula, one of the following is true for the model checking problem: (i) it is $\NL$-complete; (ii) it is $\LOGSPACE$-complete; (iii) it is first-order definable (hence, in $\ACzero$).
  Furthermore, we classify the complexity of the model checking problem for disjunctions of two arbitrary dependence atoms, and also characterize when such a disjunction is coherent,  i.e., when it satisfies a certain  small-model property.
  Along the way, we identify a new class of 2CNF-formulas whose satisfiability problem is $\LOGSPACE$-complete.
\end{abstract}

\section{Introduction}

Notions of dependence and independence are ubiquitous in several different areas of mathematics and computer science, including linear algebra, probability theory, and database theory. 
Dependence logic $\DL$, developed by  V\"a\"an\"anen in \cite{DBLP:books/daglib/Jouko},
is a formalism for expressing and studying  such notions. The basic building blocks of dependence logic are the \emph{dependence atoms}, which, in effect, express functional dependencies on database relations.  Recall that functional dependencies
are integrity constraints on databases asserting  that the values of some attribute of a database relation is determined by the values of some of the other attributes of that relation (see \cite{DBLP:books/aw/AbiteboulHV95}). For example, if $R(\course,\dpt,\inst)$ is a database relation with attributes for courses, departments, and instructors, then the functional dependency $(\course \rightarrow \dpt)$ asserts that each course is offered by only one department, while the functional dependency $(\inst \rightarrow \course)$ asserts that each  instructor teaches only one course.
In dependence logic $\DL$, these two functional dependencies are captured by the dependence atoms $\dep(\course,\inst)$ and
$\dep(\inst,\course)$, respectively.  
The formulas of $\DL$ are obtained by combining 
dependence atoms,  atomic formulas, and negated  atomic formulas using conjunction $\wedge$, disjunction $\lor$,  existential quantification $\exists$, and universal quantification~$\forall$.

   Even though the syntax of $\DL$ resembles that of first-order logic $\FO$, the semantics of $\DL$ is second-order. Instead of single assignments of values to variables, 
   the semantics of $\DL$ uses  sets of assignments, called \emph{teams}, that represent  database relations on which the dependence atoms are interpreted. The second-order character of the semantics of $\DL$ is already manifested in the semantics of disjunction:
    a structure $\struct$ and a team $T$ satisfy an $\DL$-formula of the form $\varphi_1 \lor \varphi_2$ (in symbols, $(\struct,T)\models \varphi_1 \lor \varphi_1$) if there are two teams $T_1$ and $T_2$ such that  $T=T_1\cup T_2$, 
   $(\struct, T_1) \models \varphi_1$,
   and $(\struct,T_2) \models \varphi_2$.
   In terms of expressive power and as regards sentences, $\DL$ is known to have the same expressive power as existential second-order logic $\ESO$~\cite{KontinenV09,DBLP:books/daglib/Jouko}. From this  result and Fagin's Theorem \cite{Fagin74}, it follows that, on the class of all finite structures, the sentences of dependence logic can express precisely all decision problems in $\NP$.  In particular, there are $\DL$-sentences that can express $\NP$-complete problems; in other words,   the model checking problem for such sentences is $\NP$-complete.

 In view of the preceding state of affairs, it is natural to ask: are there syntactic restrictions such that   the model checking problem for $\DL$-formulas obeying these restrictions  is tractable? Rather surprisingly, Kontinen \cite{JaKontinen13} showed that the model checking problem can be $\NP$-complete
 even for  disjunctions of three  dependence atoms.
Concretely, consider the $\DL$-formula
$\dep(x,y)\lor \dep(u,v) \lor \dep (u,v)$. In \cite{JaKontinen13}, it is shown that the following problem is $\NP$-complete: given a team $T$ of arity $4$, does $T\models \dep(x,y)\lor \dep(u,v) \lor \dep (u,v)$?
In other words, are there three teams $T_1,T_2,T_3$ such that $T=T_1 \cup T_2 \cup T_3$, $T_1\models \dep(x,y)$, $T_2\models \dep(u,v)$, and  $T_3\models \dep(u,v)$? On the tractability side, Kontinen \cite{JaKontinen13}
showed that the model checking problem for disjunctions of  two dependence atoms is always in $\NL$, i.e., it is always solvable in non-deterministic logarithmic space. This was achieved via a logspace reduction of the model checking problem for a  disjunction of two dependence atoms to $\TwoSAT$, the satisfiability problem for $2$-CNF formulas, which is well known to be in $\NL$.

\subparagraph*{Contributions.}
In this paper, we carry out a systematic investigation of the disjunctions of two dependence atoms. The motivation for this investigation is twofold. First, while the model checking for disjunctions of two dependence atoms is always in $\NL$, there are such  disjunctions for which the model checking is $\NL$-complete, while for others it is $\FO$-definable (hence also in uniform $\ACzero$).  So, one main aim is to pinpoint the exact complexity of the model checking problem for a given disjunction of two dependence atoms.
Second, disjunctions of dependence atoms capture natural database integrity constraints that have not been considered earlier. To make  this point, let us consider the disjunction  $\dep(\course,\dpt) \lor \dep(\inst,\dpt)$. The disjunct $\dep(\course,\dpt)$
expresses the functional dependency that each course is offered by only one department, while the disjunct $\dep(\inst,\dpt)$ expresses the functional dependence that each instructor is affiliated with only one department.
It is unlikely that either of these two functional dependencies holds in any university, since it is typically the case that some courses are cross-listed by two departments and some instructors have joint appointments in two or more departments. Now, the disjunction $\dep(\course,\dpt) \lor \dep(\inst,\dpt)$ expresses the constraint that the database relation $R(\course,\dpt,\inst)$ can be split into two parts such that the functional dependency $\dep(\course,\dpt)$ holds in the first part and the functional dependency $\dep(\inst,\dpt)$ holds in the second part. This is a more relaxed integrity constraint, and it is  quite plausible that it holds in many universities. Yet, this type of integrity constraint has not been investigated in database theory. The reason is that the largest collection of integrity constraints in databases studied in the past is the collection of \emph{embedded implicational dependencies} \cite{FV86}, each of which is $\FO$-definable, while, as our result  will imply,  the disjunction 
$\dep(\course,\dpt) \lor \dep(\inst,\dpt)$ in not $\FO$-definable.
We believe that disjunctions of dependence atoms deserve to be studied as database integrity constraints in their own right;  the work reported here makes a first step in this direction by focusing on the model checking problem for such formulas.

We classify the complexity of the model checking problem for disjunctions of two dependence atoms by establishing a trichotomy theorem: for every $\DL$-formula $\varphi$ that is the disjunction of two dependence atoms, one of the
following three statements holds: (i) the model checking problem for $\varphi$ is $\NL$-complete;
(ii) the model checking problem for $\varphi$ is $\LOGSPACE$-complete, where $\LOGSPACE$ is the class of all decision problems solvable in deterministic logarithmic space; (iii) the model checking problem for $\varphi$ is $\FO$-definable, hence  in uniform $\ACzero$, where $\ACzero$ is the class of all decision problems solvable by constant-depth, polynomial-size circuits. 

As an illustration of this trichotomy theorem, consider again the database relation 
$R(\course,\dpt,\inst)$ with data about courses, departments, and instructors. Consider also the dependence atoms $\dep(\course,\inst)$,
$\dep(\course,\dpt)$,
$\dep(\inst,\dpt)$,
$\dep(\inst,\course)$ that express, respectively, the functional dependencies
$(\course \rightarrow \inst)$,
$(\course \rightarrow \dpt)$,
$(\inst \rightarrow \dpt)$,
$(\inst \rightarrow \course)$. Our trichotomy theorem classifies the complexity of the model checking for all possible disjunctions of these four dependence atoms. In particular, it implies that the following statements are true:
\begin{itemize}
\item The model checking problem for 
$\dep(\course,\dpt) \lor \dep(\inst,\dpt)$ is $\NL$-complete.
\item The model checking problem for 
$\dep(\inst,\course) \lor \dep(\course,\inst)$ is
$\LOGSPACE$-complete.
\item The model checking problem for 
$\dep(\course,\dpt) \lor \dep(\course,\inst)$ is $\FO$-definable, hence it is in uniform $\ACzero$.
\end{itemize}
As a byproduct of the preceding complexity-theoretic classification, we also  identify a  collection of $2$-CNF formulas for which the satisfiability problem is  $\LOGSPACE$-complete. 

The preceding results are about \emph{unary} dependence atoms $\dep(x,y)$, where $x$ and $y$ are variables. The syntax of dependence logic allows also for 
\emph{higher-arity} dependence atoms, i.e., for dependence atoms of the form $\dep(\bar x, y)$ with $|\bar x|>1$; such atoms assert that the value of the variable $y$ is determined by the values of the variables of $\bar x$. 
We leverage out results about the model checking problem for disjunctions of two unary dependence atoms to establish an analogous trichotomy theorem about the model checking problem for disjunctions of two higher-arity dependence atoms.
%: this problem is either $\NL$-complete or $\FO$-definable.

We complement these complexity-theoretic classifications of the model checking problem  with a structural classification of disjunctions of two dependence atoms.
Specifically, we determine whether a given disjunction of two dependence atoms is a \emph{coherent}  or an \emph{incoherent} $\DL$-formula. 
As defined in \cite{JaKontinen13}, a $\DL$-formula $\varphi$ is \emph{coherent} if there is an integer $k$ such that a team $T$ satisfies $\varphi$ if and only if every sub-team $T'$ of $T$ of size $k$ satisfies $\varphi$. 
If no such $k$ exists, then $\varphi$ is \emph{incoherent}. 
We show that a disjunction of two dependence atoms is coherent if and only if the model checking for this formula is $\FO$-definable. 
Furthermore, for each coherent disjunction of two dependence atoms, we determine its
\emph{coherence level}, i.e., the 
smallest $k$ that establishes the coherence of 
that disjunction.

The work reported here paves the way for further 
exploration of the model checking problem for the quantifier-free fragment of dependence logic and, perhaps more importantly, it paves the way for further
interaction between dependence logic and database theory.

Our findings for disjunctions of two unary dependence atoms are summarized in Table \ref{tab:current results}. 

\begin{table}
    \centering
    \begin{tabular}{cll}
    \toprule
        Formula & Coherence& model checking\\ 
    \midrule    
        $\dep(x, y) \lor \dep(z, u)$ & incoherent \rpointer{prop:dep(x y) lor dep(z u) incoherent}{Prop.~} & $\NL$-complete \rpointer{prop:lor of 2 coherent in NL}{Prop.~} \\
        $\dep(x, z) \lor \dep(y, z)$ & incoherent \rpointer{thm:dep(x z) lor dep(y z) incoherent}{Thm.~} & $\NL$-complete \rpointer{thm:dep(x z) lor dep(y z) NL-complete}{Thm.~} \\  
        $\dep(x, y) \lor \dep(y, z)$ & incoherent \rpointer{cor:dep(x y) lor dep(y z) incoherent}{Thm.~} & $\NL$-complete \rpointer{thm:dep(x y) lor dep(y z) NL-complete}{Thm.~} \\
    \midrule
        $\dep(x, y) \lor \dep(y, x)$ & incoherent \rpointer{thm:dep(x y) lor dep(y x) incoherent}{Thm.~} & $\LOGSPACE$-complete \rpointer{thm:dep(x y) lor dep(y x) in L}{Thm.~} \\
    \midrule
        $\dep(x, y) \lor \dep(x, y)$ & coherent (level 3) \rpointer{prop:bigvee dep(x y) k-coherent}{Prop.~} & $\FO$ \rpointer{coherent-FO-reduction}{Prop.~} \\
        $\dep(x, y) \lor \dep(x, z)$ & coherent (level 4) \rpointer{thm:dep(x y) lor dep(x z) 4-coherent}{Thm.~} & $\FO$ \rpointer{coherent-FO-reduction}{Prop.~} \\
    \bottomrule
    \end{tabular}
    \caption{Complexity and coherence of disjunctions of two unary dependence atoms.}
    \label{tab:current results}
\end{table}

\section{Preliminaries}
This section contains some basic material about dependence logic. For additional material, we refer the reader to the monograph \cite{DBLP:books/daglib/Jouko}.

A \emph{vocabulary} is a tuple $\tau=(R_1,\ldots,R_m)$   of relation symbols, each of which has a specified natural number $r_i$ as its arity. A $\tau$-\emph{structure} is a tuple $\calA = (A,R_1^\calA,\ldots,R_m^\calA)$, where $A$ is a set, called the \emph{domain} of $\calA$, and each $R_i^\calA$ is an $r_i$-ary relation on $A$. 

Let  $\VAR$ be a countably infinite set of variables. The syntax of dependence logic $\DL$ over a vocabulary $\tau$ extends the syntax of first-order logic $\FO$  in negation normal form with dependence atoms as additional atomic formulas.  A \emph{dependence atom} is an expression of the form $\dep(\bar y, z)$,
where $\bar y = y_1,\ldots,y_n$ is a collection of distinct variables, with $n \geq 0$.
The right-hand side $z$ is a variable that does not occur in $\bar y$.
We will write $\bar yu$ for the collection $y_1, \dots, y_n, u$ and implicitly assume that $u$ did not already appear in $\bar y$.
If $n=1$, we say that the dependence atom is \emph{unary}.
Intuitively, a dependence atom asserts that the value of the variable $z$ depends on the values of the variables $y_1,\ldots,y_n$. In particular, if $n=0$, we have a dependence atom of the form $\dep(z)$, which asserts that $z$ is a constant value, since it depends on no other variables; such atoms are called \emph{constancy} atoms.

The formulas of dependence logic $\DL$ are defined by the following Backus-Naur form:
\[	\psi {\coloneqq}\,
	x_k {=}x_\ell\mid 
	R_i(\bar x_i)\mid
	\lnot R_i(\bar x_i)\mid
    \dep(\bar y, z) \mid
	\psi{\land}\psi\mid
	\psi{\lor}\psi\mid
	\exists x\psi\mid
	\forall x\psi,
\]
where  the following hold:
\begin{itemize}
    \item 
$R_i$ is a relation symbol in $\tau$ of arity $r_i\in\mathbb N$;
\item $\bar x_i = x_1x_2\cdots,x_{r_i}$ and $x_k,x_\ell,x_1,\ldots,x_{r_i}$ are (not necessarily distinct) variables in $\VAR$;
\item 
$\bar y = y_1y_2\cdots y_n$ and $y_1,y_2,\ldots,y_n,z$ are distinct variables in $\VAR$.
\end{itemize}

Recall that the semantics of $\FO$-formulas are given using a structure $\calA$ and an assignment, i.e., a mapping from  $\VAR$ to the domain of $\calA$. Due to the  presence of dependence atoms, however, the semantics of $\DL$-formulas are given using a structure $\calA$ and a set of assignments. 

\begin{definition}
Let $X$ be a set of variables (i.e., $X\subseteq \VAR$) and let $A$ be a set.
\begin{itemize}
\item An \emph{assignment with domain $X$ and range $A$}
is a mapping $s\colon X \to A$.

\item A \emph{team with domain $X$ and range $A$} is a set $T$ of assignment with domain $X$ and range $A$.
We will write $\Dom(T)$ to denote the domain $X$ of the team $T$.
Similar $\Range(T)$ denotes the range $A$ of the team $T$.
\item If $T$ is a team with domain $X$ and  $Y\subseteq X$, then the
\emph{restriction of  $T$ on $Y$} is the team  $T\upharpoonright Y = \{\,s\upharpoonright Y \mid s\in T\,\}$, where
$s\upharpoonright Y$ is the restriction of the assignment $s$ to $Y$.
\end{itemize}

\end{definition}

Before giving the semantics of $\DL$-formulas, we need to introduce two operations on teams.

\begin{definition}
Let $\VAR$ be the set of all variables and let $A$ be a set.
\begin{itemize}
\item If $s\colon \VAR  \to A$ is an assignment with domain $\VAR$ and range $A$, $x$ is a variable, and $a$ is an element of $A$, then $s^x_a\colon \VAR \to A$ is the assignment 
such that $s^x_a(x) = a$
and $s^x_a(y)=s(y)$, for every variable $y\not = x$.

\item If $T$ is a team with domain $\VAR$ and range $A$,  $f\colon T \to A$ is a function, and $x$ is a variable, then the \emph{supplement team} $T^x_f$ is the team 
$ \{s^x_{f(s)} \mid s \in T\}.$

\item If   $T$ is a team with domain $\VAR$ and range $A$, and $x$ is a variable, then the \emph{duplicate team} $T^x_A$ is the team  $ \{s^x_a \mid x \in \VAR ~  \mbox{{\rm and}} ~  a \in A \}$.
\end{itemize}
\end{definition}
The notion of supplement (and similarly duplicate) teams can easily be extend to collections of variables.
Denote by $T_f^{\bar x}$ the supplement team $((T_f^{x_1})_f^{x_2}\cdots)_f^{x_n}$, where $\bar x = x_1, \dots, x_n$.

We are now ready to define the semantics of dependence logic $\DL$. In doing so, the supplement team will be used to define the semantics of existential quantification, while the duplicate team will be used to define  the semantics of universal quantification.
Furthermore, we extend the notion of assignments $s$ to collections $\bar y = y_1, \dots, y_n$ of variables and write $s(\bar y)$ for the tuple $(s(y_1), \dots, s(y_n))$ and $s(\bar y) = s'(\bar y)$, if $s(y_i) = s'(y_i)$ for all $1 \leq i \leq n$.

\begin{definition}\label{def-semantics}
	Let $\tau=(R_1,\ldots,R_m) $ be a vocabulary and let $\calA=(A,R_1^\calA,\ldots,R_m^\calA)$ be  a $\tau$-structure.
    The satisfaction relation $(\calA,T) \models \varphi$, where $T$ is a team with domain $\VAR$ and range $A$, and $\varphi$ is a formula of dependence logic $\DL$,  is defined by induction on the construction of $\varphi$ and simultaneously for all teams with domain $\VAR$ and range $A$ as follows:
\begin{itemize}

    \item $(\calA,T)\models x_k= x_l$   if for all  $s\in T$, we have that $s(x_k) = s(x_l)$;
    \item $(\calA,T)\models R_i(\bar x)$ if for all  $s\in T$, we have that $s(\bar x) \in R_i^{\calA}$;
    \item $(\calA,T)\models \neg R_i(\bar x)$ if for all  $s\in T$, we have that  $s(\bar x) \not\in R_i^{\calA}$;
    \item $(\calA,T)\models \dep(\bar y,z)$ if for all $s_1,s_2\in T$ with $s_1(\bar y) = s_2(\bar y)$, we have that $s_1(z)=s_2(z)$;
    \item $(\calA,T)\models \varphi_1\land \varphi_2$   if  $(\calA,T)\models \varphi_1$  and  $(\calA,T)\models \varphi_2$; 
    \item $(\calA,T)\models \varphi_1\lor \varphi_2$   if there are teams  $T_1$, $T_2$ with $T=  T_1\cup T_2$ and $ (\calA,T_i)\models \varphi_i$,  $i =1,2$;
    \item $(\calA,T)\models\exists x\varphi$  if 
    there is a function 
    $f\colon T\to A$ such that $(\calA,T^x_f)\models\varphi$, where $T^x_f$ is the supplement team associated with $T$, $x$, and $f$;
\item $(\calA,T)\models\forall x\varphi$  if  $(\calA,T^x_A)\models\varphi$, where $T^x_A$ is the duplicate team associated with $T$ and $x$.
\end{itemize}
\end{definition}

The next proposition, whose proof can be found in \cite{DBLP:books/daglib/Jouko}, states two important properties of dependence logic. 
In what follows, if $\varphi$ is a $\DL$-formula, then $\Fr(\varphi)$ is the (finite) set of all free variables of $\varphi$.
 
\begin{proposition}[{\cite[Lemma~3.27, Prop.~3.10]{DBLP:books/daglib/Jouko}}] \label{lem:properties}
Let $\calA=(A,R_1^\calA,\ldots,R_m^\calA)$ be a $\tau$-structure and let $\varphi$ be a $\DL$-formula.
\begin{description}
    \item [\emph{(Locality})] Let $X$ be a set  such that $\Fr(\varphi)\subseteq X \subseteq \VAR$.
    If $T_1$, $T_2$ are teams with domain $\VAR$ and range $A$ such that  $T_1\upharpoonright X = T_2\upharpoonright X$, then
    $(\calA,T_1) \models \varphi$  if and only if  $(\calA,T_2) \models \varphi$.

    \item [\emph{(Downward Closure)}] Let $P$ and $T$ be teams with domain $\VAR$ and range $A$ such that $P\subseteq T$. If
    $(\calA,T)\models \varphi$, then also $(\calA,P)\models \varphi$.
\end{description}    
\end{proposition}

The Locality property tells that the satisfaction relation $(\calA, T) \models \varphi$ depends only on the restriction of the team $T$  on the set   of  the free variables of $\varphi$. In particular, if $\psi$ is a $\DL$-sentence (i.e., a $\DL$-formula with no free variables), then
$(\calA,T)\models \psi$ if and only $\calA,\{\emptyset\})\models \psi$, where $\{\emptyset\}$ is the team consisting of the empty assignment.

V\"a\"an\"anen \cite{DBLP:books/daglib/Jouko} showed that $\DL$-sentences have the same expressive power as the sentences of  existential second-order logic $\ESO$. Combined with  Fagin's theorem \cite{Fagin74}, this result implies that, on the class of all finite structures, $\DL$-sentences express precisely all $\NP$ problems. 
Kontinen and V\"a\"an\"anen \cite{KontinenV13} showed that $\DL$-formulas (with free variables) have the same expressive power as the \emph{downward closed} formulas of existential second-order logic $\ESO$ (see~\cite{KontinenV13} for the precise statement of this result).

In view of the Locality property, from now on we  will only consider teams on finite domains, as long as, in each case, the domain of the team considered contains the free variables of the $\DL$-formula at hand. 

\subparagraph*{Coherence.}
In the following, we present the definition of a small-model property for dependence logic.
    \begin{definition}[$k$-coherence, {\cite[Def.~3.1]{JaKontinen13}}]
    Let  $\varphi(x_1, \dots, x_n)$ be  a quantifier-free $\mathcal{D}$-formula.
    We say that  $\varphi$ is \emph{$k$-coherent} if for all structures $\struct$ and teams $T$ of range $\mathcal{M}$ such that $\free(\varphi) \subseteq \Dom(T)$, the following are equivalent:
    \begin{enumerate}[(1.)]
        \item $(\mathcal{M}, T) \models \varphi$.
        \item For all $k$-element sub-teams $S \subseteq T$ it holds that $(\mathcal{M}, S) \models \varphi$.
    \end{enumerate}
\end{definition}

Notice that for the above equivalence, the direction ``\itemstyle{(1.)} $\Rightarrow$ \itemstyle{(2.)}'' is always true due to $\DL$ being downwards closed (see~\Cref{lem:properties}).

\begin{definition}[coherence-level, incoherence]
The \emph{coherence-level} of $\varphi$ is the least natural number $k \in \mathbb{N}$ such that $\varphi$ is $k$-coherent; if such a $k$ does not exist, we call $\varphi$ \emph{incoherent}.
\end{definition}

Clearly, first-order atomic formulas have a coherence-level of $1$ and dependence atoms a coherence-level of $2$. 

\begin{proposition}\label{prop:bigvee dep(x y) k-coherent}\label{prop:dep(x y) lor dep(z u) incoherent}
The following statements are true:
\begin{enumerate}[(1.)] 
    \item The formula $\dep(x, y) \lor \dep(x, y)$ has coherence-level $3$;  \hfill\cite[Prop.~3.8]{JaKontinen13} 
    \item The formula $\dep(x, y) \lor \dep(z, u)$ is incoherent. \hfill\cite[Thm.~3.11]{JaKontinen13}
\end{enumerate}
\end{proposition}

 We will use the shortcut $\Rel(T) \coloneqq \{(s(x_1), \dots, s(x_n)) \mid s \in T\}$, when the team $T$ is given as a relation. 
 The following proposition shows that every quantifier-free $k$-coherent formula of dependence logic is first-order definable (and, hence, is in uniform $\ACzero$).
\begin{proposition}[{\cite[Thm.~4.9]{JaKontinen13}}]\label{coherent-FO-reduction}
    Suppose  $\varphi(x_1, \dots, x_n)\in\DL$ is a quantifier-free $k$-coherent formula over a vocabulary $\tau$.
    Then there is a sentence $\varphi^* \in \FO(\tau \cup \{R\})$, where $R$ is $n$-ary, such that for all $\tau$-structures $\mathcal M$ and for all teams $T$ of domain $\{x_1, \dots, x_n\}$, we have that
    $(\mathcal M, T) \models \varphi(x_1, \dots, x_n)$ 
   if and only if 
    $\Rel(T) \models \varphi^*(R)$.
\end{proposition}

\section{Model checking}\label{sec:mc}
The  model checking problem is a fundamental decision problem arising in every logic. Informally, this problem asks if a finite structure satisfies a formula of the logic at hand. Actually, according to the taxonomy introduced by Vardi \cite{DBLP:conf/stoc/Vardi82}, there are two versions of the model checking problem, the \emph{combined complexity} version and the \emph{data complexity} version. In the first version, 
the input consists of a formula and a finite structure, 
 while in the second version the formula is fixed and the input consists of just a finite structure.

Here, we focus on the data complexity version of the model checking problem for dependence logic $\DL$. Specifically, every $\DL$-formula $\varphi$ gives rise to the following decision problem.

\problemdef{$\MC(\varphi)$ --- the model checking problem for $\varphi\in\DL$}{A finite $\tau$-structure $\calA=(A,R_1^\calA,\ldots, R_m^\calA)$ and a team $T$  with domain a finite set $X$ and range $A$ such that $\free(\varphi) \subseteq X \subseteq \VAR$}
{Does $(\struct,T)\models\varphi$}

As mentioned in the introduction, the results in \cite{KontinenLV13,DBLP:books/daglib/Jouko} imply that $\MC(\varphi)$ is in $\NP$ for every $\DL$-formula $\varphi$, and  that there are $\DL$-formulas $\psi$ for which $\MC(\psi)$ is an $\NP$-complete problem.  Furthermore,  the intractability of the model checking problem is even true for some disjunctions of three dependence atoms; for example,  $\MC(\dep(x,y)\lor \dep(u,v)\lor \dep(u,v))$ is $\NP$-complete, as shown by Kontinen~\cite{JaKontinen13}.

Note that if $\varphi$ is
a $\DL$-formula in which no relation symbols from the vocabulary $\tau$-occur, then the satisfaction relation $(\calA,T)\models \varphi$ depends only on the team $T$ and not on the $\tau$-structure $\calA$. 
From this point on, we will focus on dependence logic formulas in which  no relation symbols from the vocabulary $\tau$ occur. For this reason, we will drop $\calA$ from the satisfaction relation and, instead, write $T\models \varphi$, where $\varphi$ contains no 
relation symbols from $\tau$. Thus, if $\varphi$ is such a $\DL$-formula, then the model checking problem $\MC(\varphi)$ asks: given a team $T$ with domain a finite set $X$ such that $\free(\varphi)\subseteq X\subseteq \VAR$, does
$T\models \varphi$?

The following result yields a sufficient condition for the model checking problem to be in $\NL$, the class of all decision problems solvable in non-deterministic logarithmic space.

\begin{restatable}[{\cite[Thm.~4.12]{JaKontinen13}}]{proposition}{tosat}\label{prop:lor of 2 coherent in NL}
    Suppose $\varphi$ and $\psi$ are $2$-coherent quantifier-free $\mathcal{D}$-formulas with no relations symbols.
    Then there is a logarithmic-space reduction from the model checking problem $\MC(\varphi \lor \psi)$ to $\TwoSAT$.
    Consequently, $\MC(\varphi \lor \psi)$ is in $\NL$.
\end{restatable}

We now describe the reduction of $\MC(\varphi \lor \psi)$ to $\TwoSAT$ in some detail, since it will be relevant later on (\Cref{sec:krom sat}) in the present paper.

    Given a team $T = \{s_1, \dots, s_k\}$, go through all two-element subsets $\{s_i, s_j\}$ of $T$, and construct  a set $C$ of two-clauses as follows:
\begin{itemize}
    \item If $\{s_i, s_j\} \not \models \varphi$, then $(x_i \lor x_j) \in C$.
    \item If $\{s_i, s_j\} \not\models \psi$, then $(\lnot x_i \lor \lnot x_j) \in C$.  
\end{itemize}
Then let $\Theta_T$ be the $2$CNF-formula $ \bigwedge_{\theta \in C} \theta$. It can be shown that $T\models \varphi 
\lor \psi$ if and only if
$\Theta_T$ is satisfiable (see \cite{JaKontinen13} for a detailed proof).

Since every dependence atom is a $2$-coherent formula, we obtain the following result.
\begin{corollary}
    If $\varphi$ and $\psi$ are dependence atoms, then $\MC(\varphi\lor\psi)$ is in $\NL$.
\end{corollary}
\subsection{NL-Completeness}
This section will cover 
the cases in which the model checking problem for disjunctions of two unary dependence atoms is complete for nondeterministic logarithmic space.
\begin{restatable}[{\cite[Thm.~4.14]{JaKontinen13}}]{proposition}{fromsat}\label{prop:dep(x y) lor dep(z w) NL-hard}
    The problem $\MC(\dep(x, y) \lor \dep(z, w))$ is $\NL$-complete.
\end{restatable}
We will outline the reduction, constructed in the proof of this result, because the next result (\Cref{thm:dep(x y) lor dep(y z) NL-complete}) is based on this construction. 
The full proof of  Proposition \ref{prop:dep(x y) lor dep(z w) NL-hard} can be found in \cite{JaKontinen13}.
The intuition of the construction is that the team encodes the clauses of the formula. 
The split ensures a consistent assignment on the left side, while the right side allows to ``buffer'' one unsatisfied literal per clause.
Consider a $\TwoSAT$ instance $\theta(p_1, \dots, p_m) = \bigwedge_{i \in I} E_i$, with $E_i = (A_{i_1} \lor A_{i_2}), i \in I$, where $A_{i_1}, A_{i_2}$ are literals.
For each conjunct $E_i$, create a team $T_{E_i} = \{s_{i_1}, s_{i_2}\}$ with domain $\{x, y, z, w\}$, where variable $x$ ranges over $\{p_1, \dots, p_m\}$; variable $y$ ranges over the truth values $\{0, 1\}$;
 variable $z$ ranges over the indices $i \in I$, thus $z$ denotes the clause $E_i$; and 
variable $w$ ranges over the values $\{0, 1\}$.
Variables $z$ and $w$ ensure that at least one of the assignments from each $T_{E_i}$ go into the subset of $T$ that eventually encodes an assignment satisfying $\theta$.
For example, the team $T_{E_i}$ for a clause $(p_k \lor \lnot p_j)$ is the one on the left in \Cref{tab:team for (p_k lor p_j)}.
Finally $T = \bigcup_{i \in I} T_{E_i}$.
\begin{figure}
    \centering
    $
    \begin{array}{ccccc}\toprule
        T_{E_i} & x & y & z & w \\ \midrule
        s_{i_1} & p_k & 1 & i & 0 \\
        s_{i_2} & p_j & 0 & i & 1\\\bottomrule
    \end{array}
    \qquad\quad
    \begin{array}{cccc}\toprule
         T_{E_i} & x & y & z \\\midrule
         s_{i_1} & p_k & i & 1_k \\
         s_{i_2} & p_j & i & 0_j \\\bottomrule
    \end{array}       
    \qquad
    \begin{array}{cccc}\toprule
        T_{E_i} & x & y & z \\\midrule
        s_{i_1} & i & p_k & 1 \\
        s_{i_2} & i & p_j & 0 \\\bottomrule
    \end{array}
    $\\$\,$
    
    \caption{Teams for $(p_k \lor \lnot p_j)$ in \Cref{prop:dep(x y) lor dep(z w) NL-hard} (left) and \Cref{thm:dep(x z) lor dep(y z) NL-complete} (middle/right).}
    \label{tab:team for (p_k lor p_j)}
\end{figure}

\begin{restatable}{theorem}{NLhardness}\label{thm:dep(x z) lor dep(y z) NL-complete}\label{thm:dep(x y) lor dep(y z) NL-complete}
    The following statements are true:
    \begin{enumerate}[(1.)] 
        \item The problem $\MC(\dep(x,z) \lor \dep(y,z))$ is $\NL$-complete.
        \item The problem $\MC(\dep(x,y) \lor \dep(y,z))$ is $\NL$-complete.
    \end{enumerate}
\end{restatable}
\begin{proof}
    Membership follows directly from \Cref{prop:lor of 2 coherent in NL}.
    We now give the proof of $\NL$-hardness for \itemstyle{(1.)}. 
    The proof works analogous for \itemstyle{(2.)}.
    
    We show $\TwoSAT \leqlogm \MC(\dep(x, z) \lor \dep(y, z))$.
    Suppose $\theta(p_1, \dots, p_m)$ is an instance of $\TwoSAT$ of the form $\bigwedge_{i \in I} E_i$, with conjuncts $E_i = (\ell_{i_1} \lor \ell_{i_2})$, where $\ell_{i_1}, \ell_{i_2}$ are different literals, i.\,e., $\ell_{i_1} \neq \ell_{i_2}$.
    This can easily be archived by replacing conjuncts of the form $(\ell \lor \ell)$ with $(\ell \lor q)$ and $(\ell \lor \lnot q)$, where $q\not\in \{p_1, \dots, p_m\}$ is a new proposition.

    We will construct a team $T$, such that the following are equivalent:
    \begin{itemize}
        \item $T \models \dep(x, z) \lor \dep(y, z)$.
        \item $\theta(p_1, \dots, p_m)$ is satisfiable.
    \end{itemize}
    For each conjunct $E_i$, we create a team $T_{E_i}$ with two assignments $s_{i_1}$ and $s_{i_2}$ of domain $\{x, y, z\}$, that encode the literals of $E_i$.
    To this end, variable $x$ ranges over the propositions $\{p_1, \dots, p_m\}$, variable $y$ ranges over clause (index) $I$, and variable $z$ ranges over $\{\,0_k, 1_k \mid 1 \leq k \leq m\,\}$, such that each proposition has its own two truth values.
    This is necessary to ensure that $z$ has two different values in each $T_{E_i}$.
    Thus, the dependence atom $\dep(y, z)$ makes sure we have to choose at least one of the assignments from each $T_{E_i}$ into the subset of $T$, that will satisfy $\dep(x, z)$, thereby encoding the assignment that satisfies $\theta$.

    Each literal $\ell_{i_j}$ in $E_i$, gives rise to one assignment.
    For example, the team $T_{E_i}$ for a clause $(p_k \lor p_j)$ is the one in the middle in \Cref{tab:team for (p_k lor p_j)}.
    
    Now $T$ is the union $\bigcup_{i \in I} T_{E_i}$.
    Suppose $\theta(p_1, \dots, p_m)$ is satisfiable.
    Then there exists an assignment $F\colon \{p_1, \dots, p_m\} \to \{0, 1\}$, that satisfies $\theta(p_1, \dots, p_m)$.
    We write $F(p_k)_k$ to denote the truth values $\{0_k, 1_k\}$ used by variable $z$. For example, if $F(p_k) = 0$, then $F(p_k)_k = 0_k$.
    Define the partition of the team $T$ into two sets in the following way:
    \begin{align*}
        T_1 &= \{s \in T \mid s(x) = p_k \text{ and } s(z) = F(p_k)_k\} \\
        T_2 &= T \setminus T_1
    \end{align*}
    The assignments in $T$ that agree with the assignment $F$ are chosen to $T_1$.
    Since $F$ evaluates $\bigwedge_{i \in I} E_i$ true, it satisfies every conjunct $E_i$.
    Thus $T_1$ contains at least one assignment from each $T_{E_i}$.
    Thus there will be at most one tuple from each $T_{E_i}$ left to $T_2$.
    This $T_2$ trivially satisfies $\dep(y, z)$ since all tuples in $T_2$ disagree on $y$, i.\,e., for all $s, t \in T_2$ it is true that $s(y) \neq t(y)$, if $s \neq t$.

    Next we will show that $T_1$ satisfies $\dep(x, z)$.
    Let $s, t \in T_1$, such that $s(x) = t(x) = p_k$.
    Then by the definition of $T_1$ it holds that $s(y) = F(p_k)_k = t(y)$ holds.
    Thus $T_1 \models \dep(x, z)$.

    The other direction:
    Suppose $T \models \dep(x, z) \lor \dep(y, z)$.
    Then there is a partition of $T$ into $T_1$ and $T_2$, such that $T_1 \models \dep(x, z)$ and $T_2 \models \dep(y, z)$.
    We will define the assignment $F \colon \{p_1, \dots, p_m\} \to \{0, 1\}$ in the following way:
    \begin{itemize}
        \item If there exists $s \in T_1$, such that $s(x) = p_k$, then $F(p_k)_k = s(z)$.
        \item Otherwise, if for all $s \in T_1$ it holds $s(x) \neq p_k$, then $F(p_k) = 1$.
    \end{itemize}
    Let us show that $F$ is a function, which satisfies $\theta(p_1, \dots, p_m)$:
    \begin{enumerate}
        \item Clearly, $\Dom(F) = \{p_1, \dots, p_m\}$ and $\Range(F) = \{0, 1\}$.
        \item $F$ is a function:
        Let $p_k \in \{p_1, \dots, p_m\}$.
        Suppose there exists $s, t \in T_1$, such that $s(x) = t(x) = p_k$ holds.
        Since $T_1 \models \dep(x, z)$ holds, it follows that $s(z) = t(z)$ holds.
        Suppose there are no $s \in T_1$, such that $s(x) = p_k$.
        Then by definition of $F$ it holds that $F(p_k) = 1$.
        \item $F$ satisfies $\theta(p_1, \dots, p_m)$:
        Note that $y$ is constant and $z$ is assigned different values by each tuple in each $T_{E_i}$.
        Thus $T_1$ contains at least one of the tuples from each $T_{E_i}$.
        Let $s \in T_{E_i}$, such that $s \in T_1$.
        Recall, that in every assignment the values $s(z)$ encodes the truth value of $s(x)$ such that $E_i$ is satisfied.
        Since $s$ agrees with $F$, it holds that $F(\ell_{i_j})_{i_j} = s(z)$, which implies that $F(E_i) = 1$.
    \end{enumerate}
    Each conjunct $E_i$, of $\theta$ gives rise to a constant size team of two assignments with domain $\{x, y, z\}$.
    Thus the team $T$ can be constructed in $\LOGSPACE$ for each $\theta$. 
\end{proof}
\subsection{L-Completeness}\label{sec:L-complete}
This section will cover 
the cases in which the model checking problem for disjunctions of two unary dependence atoms is complete for logarithmic space.

If  $T$ is a team with domain $X$ and $x$ is a variable in $X$, 
then we will use the notation $\rngrestr{T}{x}\coloneqq\Range(T\upharpoonright\{x\})$ to denote the set of values in the range of $T$ that $x$ takes.

\begin{definition}\label{def:G_T}
    Let $T$ be a team with domain $\{x,y\}$ such that $\rngrestr{T}{x} \cap \rngrestr{T}{y} = \emptyset$.
    The \emph{undirected graph of $T$} is $G_T \coloneqq (V_T, E_T)$, where 
    \begin{align*}
        V_T &= \Range(T), \qquad E_T = \bigl\{\,\{s(x),s(y)\}\,\big|\, s\in T\,\bigr\}.
    \end{align*}
\end{definition}
Notice that, by definition, $G_T$ is a bipartite graph. 
The next lemma provides a graph-based interpretation of the condition under which a team $T$ satisfies the formula $\dep(x,y)\lor \dep(y,x)$.%
\begin{restatable}{lemma}{satIffDirected}\label{lem:sat iff directed}
    Let $T$ be a team with domain $\{x,y\}$ such that $\rngrestr{T}{x} \cap \rngrestr{T}{y} = \emptyset$.
    The following statements are equivalent:
    \begin{enumerate}[(1.)]
        \item $T \models \dep(x,y) \lor \dep(y,x)$.
        \item $G_T$ can be transformed into a directed graph $G'_T = (V_T, E'_T)$ such that 
        \begin{enumerate}[(a)]
            \item each node has at most one out-going edge, and 
            \item for each $\{u, v\} \in E_T$ either $(u,v) \in E'_T$ or $(v,u) \in E'_T$.
        \end{enumerate}
        \item For each connected component $C = (V_C, E_C)$ of $G_T$, we have that $|E_C| \le |V_C|$. 
    \end{enumerate}    
\end{restatable}
\begin{proof}
    We will prove the equivalences as follows. 
    \begin{description}
        \item[(1.) $\Rightarrow$ (2.)] Assume, w.l.o.g., $T = T_1 \uplus T_2$, such that $T_1 \models \dep(x,y)$ and $T_2 \models \dep(y,x)$.
        Then the directed graph $G'_T = (V_T, E'_T)$ with 
        \[
            E'_T = \{\,(s(x),s(y)) \mid s \in T_1\,\} \cup \{\,(s(y),s(x)) \mid s \in T_2\,\}
        \]
        has properties (a) and (b). 

        For (a), wrongly assume, w.l.o.g., that there are two edges $e_1=(a,b)$ and $e_2=(a,b')$ with $b\neq b'$ and $a,b,b'$ values in $\Range(T)$ (the same argument could be made for the second component being $a$ and the first components being different). 
        Now, these edges correspond to assignments $s(x)=a$, $s(y)=b$, and $s'(x)=a$, $s'(x)=b'$. 
        Both $s,s'\in T_1$ and would violate $\dep(x,y)$ which is a contradiction.

        For (b), notice that each $\{u,v\}$ corresponds to a single assignment $s\in T$. 
        By assumption, $s$ is either in $T_1$ or $T_2$, whence it has to appear in $E'_T$ in either direction.
        \item[(2.) $\Rightarrow$ (1.)] Let $G'_T = (V_T, E'_T)$ be a directed graph, where each node has at most one out-going edge.
        Define $T_1 = \{s \in T \mid (s(x), s(y)) \in E'_T\}$ and $T_2 = \{s \in T \mid (s(y), s(x)) \in E'_T\}$.
        Then, $T_1 \models \dep(x, y)$, because each $s(x)$ appears only once by assumption, so the dependency is not violated.
        By analogous arguments we have that $T_2 \models \dep(y,x)$.
        \item[(2.) $\Rightarrow$ (3.)] We make a case distinction.
        If each node in $G'_T$ has exactly one out-going edge, then $|V_T| = |E'_T| \overset{\text{\tiny(b)}}{=} |E_T|$ and likewise $|E_C| = |V_C|$ for each connected component $C$ of $G_T$.
        Otherwise, if some nodes have no out-going edge, then $|E_C| < |V_C|$.
        \item[(3.) $\Rightarrow$ (2.)] A connected component in $G_T$ has at least $|V_C| - 1$ edges (otherwise it would not be connected).
        So there are only two cases:
        \begin{enumerate}[(i)]
            \item \label{enum:cc is tree} If $|E_C| = |V_C| - 1$, then $C$ is a (spanning-)tree. 
            Define $G'_T$ by orienting the edges from the leaves to the root, i.e., bottom-up.
            Since each node has one parent, each node has one outgoing edge (except for the root which has no outgoing edge).
            \item If $|E_C| = |V_C|$, then $C$ contains \emph{exactly one} cycle $\mathcal C$, because $C$ is connected and has one more edge than its spanning-tree.
            Choose any orientation of $\mathcal C$ (clockwise or counter-clockwise).
            Afterwards, all nodes in $C\setminus\mathcal C$ (hence outside of the cycle $\mathcal C$) form tree-like connected components and we can proceed as in {\sffamily\bfseries(\ref{enum:cc is tree})}: 
            There is one node directly connecting to $\mathcal C$ which will become the root, and its outgoing edge connects to a node in the cycle $\mathcal C$.\qedhere
        \end{enumerate}
    \end{description}
\end{proof}
A visual example for $G_T$ and the equivalence of \itemstyle{(1.)} and \itemstyle{(2.)} are depicted in \Cref{fig:ex:G_T}.
\begin{figure}
    \centering
    $
    \begin{array}{ccc}\toprule
        T & x & y \\\midrule
        s_1 & 0 & \hat 0 \\
        s_2 & 0 & \hat 1 \\
        s_3 & 1 & \hat 0 \\
        s_4 & 1 & \hat 1 \\\bottomrule
    \end{array}
    $\quad$\leftrightsquigarrow$\quad
    \begin{tikzpicture}[scale=2, baseline=(current bounding box.center)]
        \node (0) at (0,0) {$0$};
        \node (1) at (0,-1) {$1$};
        \node (h0) at (1,0) {$\hat 0$};
        \node (h1) at (1,-1) {$\hat 1$};

        \draw (0) -- node[above] {$s_1$} (h0);
        \draw (0) -- node[near start, left] {$s_2$} (h1);
        \draw (1) -- node[near start, left] {$s_3$} (h0);
        \draw (1) -- node[below] {$s_4$} (h1);
    \end{tikzpicture}
    \quad$\leftrightsquigarrow$\quad
    \begin{tikzpicture}[scale=2, baseline=(current bounding box.center)]
        \node (0) at (0,0) {0};
        \node (1) at (0,-1) {1};
        \node (h0) at (1,0) {$\hat 0$};
        \node (h1) at (1,-1) {$\hat 1$};

        \draw[-stealth'] (0) -- node[above] {$s_1$} (h0);
        \draw[stealth'-] (0) -- node[near start, left] {$s_2$} (h1);
        \draw[stealth'-] (1) -- node[near start, left] {$s_3$} (h0);
        \draw[-stealth'] (1) -- node[below] {$s_4$} (h1);
    \end{tikzpicture}
    \quad$\leftrightsquigarrow$\quad
    $
    \begin{array}{ccc}\toprule
        T & x & y \\\midrule
        s_1 & 0 & \hat 0 \\
        s_4 & 1 & \hat 1 \\\hdashline
        s_2 & 0 & \hat 1 \\
        s_3 & 1 & \hat 0 \\\bottomrule
    \end{array}
    $
    \caption{Example depicting the construction of $G_T$ from team $T$ on the left. 
    Followed by the transformation into directed graph $G'_T$, where each node has at most one out-going edge.
    This induces the split $\{s_1, s_4\} \models \dep(x, y)$ and $\{s_2, s_3\} \models \dep(y, x)$ of $T$ shown on the right.}
    \label{fig:ex:G_T}
\end{figure}

\begin{theorem}\label{thm:dep(x y) lor dep(y x) in L}
   The problem $\MC(\dep(x,y) \lor \dep(y,x))$ is in $\LOGSPACE$.
\end{theorem}
\begin{proof}
    \Cref{alg: dep(xy) lor dep(yx)} decides $T \models \dep(x,y) \lor \dep(y,x)$ requiring only logarithmic space.

    We assume the input to the algorithm to be the graph $G_T$, which is possible since its edge set is, depending on the used encoding, the same as the team $T$. 
    The correctness follows immediately from \Cref{lem:sat iff directed}, i.e., \Cref{alg: dep(xy) lor dep(yx)} returns \texttt{true} if and only if $|E_C| \le |V_C|$ for all connected components $C$ of $G_T$.
    The algorithm iterates over all nodes $v$ in $G_T$ and checks that the size of the corresponding connected component (so its number of nodes) is not smaller than its number of edges. 
    The size count, $c$, is achieved by counting the number of reachable nodes. In parallel, the sizes of their respective neighborhoods are added up in a different counter, $d$.
    Since this counts each edge twice, halving $d$ and comparing to $c$ suffices.

    The amount of space required by the algorithm is as follows. 
    Computing the number of nodes in a connected component can be done in logarithmic space in $|G_T|$ (which is in $O(|T|)$) using undirected reachability (which is well-known to be in $\LOGSPACE$~by Reingold's theorem~\cite{DBLP:journals/jacm/Reingold08}).
    The algorithm stores an index of the current node $v$ and an index of an auxiliary node $u$, hence both are, by binary encoding, in $O(\log|T|)$.
    Furthermore, two counters are used: one to count the size of the connected component, $c$, and another to count the number of edges in the connected component, $d$. 
    Both counters are, again, in $O(\log |T|)$.
    Storing these values together requires only logspace in the input size.
    The claimed result applies, as $\LOGSPACE^\LOGSPACE = \LOGSPACE$ (that is to say, $\LOGSPACE$ is low for itself).
\end{proof}
\begin{algorithm}[t]
    \DontPrintSemicolon
    \caption{model checking $\dep(x,y) \lor \dep(y,x)$}
    \label{alg: dep(xy) lor dep(yx)}
    \SetKwInOut{Input}{input}
    \Input{Graph $G_T$\tcp*[f]{Space}}
    \ForEach(\tcp*[f]{$\log |V|$}){$v \in V_T$}{
        $c \gets$ 0 \tcp*{$\log |V|$}
        $d \gets $ number of incident edges of $v$ \tcp*{$\log |E|$}
        \ForEach(\tcp*[f]{$\log |V|$}){$u \in V_T\setminus\{v\}$ and $u$ is reachable by $v$}{
            $c \gets c + 1$, 
            $d \gets d\ + $ number of incident edges of $u$\;
        }
        \lIf(\tcp*[f]{$d$ counts the edges twice}){$\frac{d}{2} > c$}{
            \Return \texttt{false}
        }
    }
    \Return \texttt{true}\;
\end{algorithm}

\begin{restatable}{lemma}{lemRangedepxydepyx}\label{lem:team to big}
Let $T$ be a team with domain $\{x,y\}$. 
    If the size $|T| > |\rngrestr{T}{x}| + |\rngrestr{T}{y}|$, then $T \not\models \dep(x, y) \lor \dep(y, x)$.    
\end{restatable}
\begin{proof}
    Assume $T \models \dep(x, y) \lor \dep(y, x)$ and $|T| > |\rngrestr{T}{x}| + |\rngrestr{T}{y}|$, then there exists a split $T = T_1 \cup T_2$ such that $T_1 \models \dep(x, y)$ and $T_2 \models \dep(y, x)$.
    Clearly, $|T_1| \le |\rngrestr{T}{y}|$ and $|T_2| \le |\rngrestr{T}{y}|$, because otherwise each of $T_1$ and $T_2$ would not satisfy the respective dependence atom. 
    That is true, because then there would be a pair $s, t \in T_1$ with $s(x) = t(x)$ and $s(y) \neq t(y)$, respectively $s, t \in T_2$ with $s(y) = t(y)$ and $s(x) \neq t(x)$.
    Furthermore, we have that $|T| \le |T_1| + |T_2|$. 
    As a result, we get $|T| \le |T_1| + |T_2| \le |\rngrestr{T}{x}| + |\rngrestr{T}{y}|$, which is a contradiction to $|T| > |\rngrestr{T}{x}| + |\rngrestr{T}{y}|$.
\end{proof}

In the following, we will utilize the definition of \emph{first-order reductions} as defined by Immerman~\cite{immerman99} and Dahlhaus~\cite{DBLP:conf/lam/Dahlhaus83}.

\newcommand{\UFA}{\ensuremath\mathrm{UFA}}
\begin{theorem}\label{thm:dep(x y) lor dep(y x) L-hard}
    The problem $\MC(\dep(x, y) \lor \dep(y, x))$ is $\LOGSPACE$-hard under first-order reductions.
\end{theorem}
\begin{proof}
    We reduce from the complement of $\UFA$ (which is known to be $\LOGSPACE$-complete~\cite{DBLP:journals/jal/CookM87}):

    \problemdef{$\UFA$ --- Undirected Forest Accessibility}{acyclic undirected graph $G = (V, E)$, nodes $u$ and $v$}{is there is a path between $u$ and $v$}

    \begin{figure}
        \centering
        \resizebox{\linewidth}{!}{$        
            \begin{array}{ccc}\toprule
               T_{\{a,b\}} & x & y \\\midrule
               s_1 & a & b \\
               s_2 & b & a \\\bottomrule
            \end{array}
            \qquad
            \begin{array}{ccc}\toprule
               T_{\{\top,u\}} & x & y \\\midrule
               s_1 & \top & u \\
               s_2 & u & \top \\\bottomrule
            \end{array}
            \qquad
            \begin{array}{ccc}\toprule
               T_{\{\top,v\}} & x & y \\\midrule
               s_1 & \top & v \\
               s_2 & v & \top \\\bottomrule
            \end{array}
            \qquad
            \begin{array}{ccc}\toprule
               T_{\{\bot,u\}} & x & y \\\midrule
               s_1 & \bot & u \\
               s_2 & u & \bot \\\bottomrule
            \end{array}
            \qquad
            \begin{array}{ccc}\toprule
               T_{\{\bot,v\}} & x & y \\\midrule
               s_1 & \bot & v \\
               s_2 & v & \bot \\\bottomrule
            \end{array}       
        $}\\$\,$
        \caption{Team $T_{\{a, b\}}$ for edge $\{a, b\} \in E$ and teams $T_{\{\top, u\}}, T_{\{\top, v\}}, T_{\{\bot, u\}}, T_{\{\bot, v\}}$.}
        \label{tab:example teams for undirected edges}
    \end{figure}
    Suppose $\langle G=(V,E), u, v \rangle$ is an instance of $\UFA$.
    We will construct a team $T$, such that $T \models \dep(x, y) \lor \dep(y, x)$ if and only if there is \emph{no path} between $u$ and $v$ in $G$.

    For each undirected edge $\{a, b\} \in E$, create a team $T_{\{a, b\}}$ with two assignments $s$ and $t$, such that $s(x) = t(y) = a$ and $s(y) = t(x) = b$.
    Next, let $\top$ and $\bot$ be two elements such that $\{\top, \bot\} \cap V = \emptyset$.
    Create four more teams $T_{\{\top, u\}}, T_{\{\top, v\}}, T_{\{\bot, u\}}$, and $T_{\{\bot, v\}}$ as depicted in \Cref{tab:example teams for undirected edges}.
    Now, $T$ is the union $\bigcup_{\{a, b\} \in E} T_{\{a, b\}} \cup \{T_{\{\top, u\}}, T_{\{\top, v\}}, T_{\{\bot, u\}}, T_{\{\bot, v\}}\}$. 
    The team $T$ can be constructed using an $\FO$ query $I\colon \STRUC[\tau_g] \to \STRUC[\tau_{2t}]$, where $\tau_g = (E^2)$ is the vocabulary of graphs with edge relation $E$ and $\tau_{2t} = (T^2)$ is the vocabulary of teams with domain of size two, encoded via the binary relation $\Rel(T)$ (see,  \Cref{coherent-FO-reduction}).
    
    Now, we will show that $\langle G=(V,E), u , v\rangle\in\overline{\UFA}$ if and only if $T\models\dep(x, y) \lor \dep(y, x)$.

    ``$\Rightarrow$'': Assume that there is no path between $u$ and $v$ in $G$.
    Then take $u$ and $v$ as roots for trees corresponding to their acyclic connected components.
    For all other connected components choose any node as root.
    We partition $T$ into two teams as follows:
    \begin{align*}
        T_1 &= \{s \in T \mid s(y) \text{ is the parent node of } s(x)\} \\
        &\cup\{s \in T \mid (s(x), s(y)) = (u, \bot)\} \cup\{s \in T \mid (s(x), s(y)) = (v, \top)\}\\
        &\cup\{s \in T \mid (s(x), s(y)) = (\bot, v)\} \cup\{s \in T \mid (s(x), s(y)) = (\top, u)\}\\
        T_2 &= T \setminus T_1 = \{s \in T \mid (s(x), s(y)) = (t(y), t(x)), t \in T_1\}
    \end{align*}
    Observe that $T_2$ contains the ``reverse directions'' of $T_1$ (see \Cref{fig:reverse directions}).
    \begin{figure}
        \centering
        \begin{tikzpicture}[scale=1]
            \node[circle, draw, minimum size = 1.2em, inner sep = 0] at (0, 0) (a) {$a$};
            \node[circle, draw, minimum size = 1.2em, inner sep = 0] at (1, 0) (b) {$b$};
            \node[circle, draw, minimum size = 1.2em, inner sep = 0] at (2, 0) (c) {$c$};
            \node[circle, draw, minimum size = 1.2em, inner sep = 0] at (1, 1) (u) {$u$};
            \node[circle, draw, minimum size = 1.2em, inner sep = 0] at (3, 0) (d) {$d$};
            \node[circle, draw, minimum size = 1.2em, inner sep = 0] at (3, 1) (v) {$v$};

            \draw (a) -- (b) -- (u) -- (c);
            \draw (d) -- (v);
        \end{tikzpicture}
        \hspace{2cm}
        \begin{tikzpicture}[scale=1]
            \node[circle, draw, minimum size = 1.2em, inner sep = 0] at (0, 0) (a) {$a$};
            \node[circle, draw, minimum size = 1.2em, inner sep = 0] at (1, 0) (b) {$b$};
            \node[circle, draw, minimum size = 1.2em, inner sep = 0] at (2, 0) (c) {$c$};
            \node[circle, draw, minimum size = 1.2em, inner sep = 0] at (1, 1) (u) {$u$};
            \node[circle, draw, minimum size = 1.2em, inner sep = 0] at (3, 0) (d) {$d$};
            \node[circle, draw, minimum size = 1.2em, inner sep = 0] at (3, 1) (v) {$v$};

            \node[circle, draw, minimum size = 1.2em, inner sep = 0] at (2, 0.75) (bot) {$\bot$};
            \node[circle, draw, minimum size = 1.2em, inner sep = 0] at (2, 1.25) (top) {$\top$};
            
            \draw[-stealth'] (a) to[out=10, in=170] (b);
            \draw[stealth'-, dashed] (a) to[out=-10, in=-170] (b);
            \draw[-stealth'] (b) to[out=100, in=-100] (u);
            \draw[stealth'-, dashed] (b) to[out=80, in=-80] (u);
            \draw[-stealth'] (c) to[out=145, in=-55] (u);
            \draw[stealth'-, dashed] (c) to[out=125, in=-35] (u);
            \draw[-stealth'] (d) to[out=100, in=-100] (v);
            \draw[stealth'-, dashed] (d) to[out=80, in=-80] (v);
            
            \draw[-stealth'] (u) to[out=-24, in=176] (bot);
            \draw[-stealth'] (bot) to[out=4, in=-156] (v);
            \draw[-stealth'] (v) to[out=176, in=-24] (top);
            \draw[-stealth'] (top) to[out=-156, in=4] (u);
            
            \draw[stealth'-, dashed] (u) to[out=-4, in=156] (bot);
            \draw[stealth'-, dashed] (bot) to[out=24, in=-176] (v);
            \draw[stealth'-, dashed] (v) to[out=156, in=-4] (top);
            \draw[stealth'-, dashed] (top) to[out=-176, in=24] (u);
        \end{tikzpicture}
        \caption{(Left): Example $\UFA$ instance. (Right): Visualization of partition into $T_1$ (solid lines) and $T_2$ (dashed lines).}
        \label{fig:reverse directions}
    \end{figure}
    Now, $s(x) \neq t(x)$ for all $s, t \in T_1$, because each $s(x) \in V$ only has one parent $s(y) \in V$ in their corresponding trees (the cases containing $\top$ or $\bot$ satisfy this condition by construction of $T_1$, as they appear only once for $x$; and $u, v$ are roots, so do not have parents themselves).
    Thus, $T_1$ trivially satisfies $\dep(x, y)$.
    Since the construction and definition of $T_2$ are symmetric, $s(x) \neq t(x)$ for all $s, t \in T_2$. This yields $T_2 \models \dep(y, x)$.

    ``$\Leftarrow$'': For the other direction, assume the contraposition, i.e., that there is a path between $u$ and $v$ in $G$.
    Let $C = (V_C, E_C)$ be the connected component with $\{u, v\} \subseteq V_C$. 
    We will focus now on the sub-team for $C$ together with the special assignments (shown in \Cref{tab:example teams for undirected edges}) and argue why this sub-team cannot be split to satisfy $\dep(x,y)\lor\dep(y,x)$ (this would imply that the full team cannot be split as well).  
    For that purpose, let $T_C$ be the union $\bigcup_{\{a,b\} \in C} T_{\{a,b\}}$ of the teams created by the edges in $C$.
    Observe that 
    \[
        |T_C| = 2|E_C| = 2(|V_C| - 1) \text{ and } |\rngrestr{T}{x}| = |\rngrestr{T}{y}| =|V_C|,
    \]
    i.e., the size of the team is twice the number of edges in $C$, because each edge adds two assignments by definition. 
    The number of edges in $C$ is then one less than the number of vertices in $C$.
    The sizes of the domains of $x$ and $y$ are the same as the number of nodes in $C$.
    
    Next, let $T_C' = T_C \cup \{T_{\{\top, u\}}, T_{\{\top, v\}}, T_{\{\bot, u\}}, T_{\{\bot, v\}}\}$, then
    \[
        |T_C'| = 2|E_C| + 8 = 2(|V_C| - 1) + 8 \text{ and } |\rngrestr{T}{x}| = |\rngrestr{T}{y}| = |V_C| + 2,
    \]
    because eight new assignments are added to the team and two new elements to the domains.
    Thus, we get
    \[
        |T_C'| = 2(|V_C| - 1) + 8= 2|V_C| + 6 > 2|V_C| + 4 = |\rngrestr{T}{x}| + |\rngrestr{T}{y}|
    \]
    and, by \Cref{lem:team to big}, we have that $T_C' \not\models \dep(x,y) \lor \dep(y,x)$.
    Therefore $T \not\models \dep(x,y) \lor \dep(y,x)$ is true, as $T_C' \subseteq T$ and due to downwards closure.
\end{proof}

\subparagraph*{A restricted case of Krom-satisfiability.}\label{sec:krom sat}
In the sequel, we will show how our results also establish the $\LOGSPACE$-completeness of the satisfiability problem restricted to a particular type of 2CNF formulas, namely, those that are monotone, transitive, and dual-free.
\begin{restatable}{definition}{specialTwoSAT}\label{def:special 2sat}
    Let $\phi$ be a 2CNF formula with variables $\Var(\phi)=\{x_1,\dots,x_n\}$, such that:
    \begin{enumerate}[(1.)]
        \item Every clause is monotone, i.e., it is of the form  $(x_i\lor x_j)$ or of the form $(\bar x_i\lor\bar x_j)$. The first type is called \emph{positive}, the second type \emph{negative}.\label{en:monotonic}
        \item No two ``dual'' clauses $(x_i\lor y_j)$, $(\bar x_i\lor\bar x_j)$ appear in  $\phi$.\label{en:nodual}
        \item Transitivity holds: if $(x_i\lor x_j)$ and $(x_j\lor x_k)$ are clauses of  $\phi$, then so is $(x_i\lor x_k)\in\phi$; furthermore,  if $(\bar x_i\lor \bar x_j)$ and $(\bar x_j\lor \bar x_k)$ are clauses of $\phi$, then so is  $(\bar x_i\lor \bar x_k)\in\phi$.\label{en:transitivity}
    \end{enumerate}
    We call such 2CNF formulas \emph{monotone transitive dual-free} and abbreviate the corresponding satisfiability problem by $\specialTwoSat$.
\end{restatable}

We call a set $\{v_1, \dots, v_k\}$ of variables in $\phi$ a \emph{(variable-)clique}, if for each pair $v_i, v_j, i \neq j$ it is true that either $(v_i \lor v_j)$ or $(\lnot v_i \lor \lnot v_j)$ is a clause in $\phi$.
Note that, by the transitivity (\ref{en:transitivity}.), the variables of a $\specialTwoSat$ instance appear in at most two cliques, once positive and once negative.
Thus we can represent $\phi$, by listing its positive and negative cliques.

\begin{restatable}{lemma}{foEqui}\label{lem:dep(x y) lor dep(y x) equivfo specialTwoSat}
    The problems $\MC(\dep(x, y) \lor \dep(y, x))$ and $\specialTwoSat$ are equivalent via first-order reductions.    
\end{restatable}
\begin{proof}
    We present reductions between $\MC(\dep(x, y) \lor \dep(y, x))$ and $\specialTwoSat$ (for illustrations, see \Cref{fig:fo reduction}).
    Let $\tau_{t} = (T^3)$ be the vocabulary of structures where $T$ encodes a team such that $T(i, j, k)$ is true if and only if $s_i(x_j) = a_k$, where $i, j, k$ are natural numbers.
    Note that this differs from the encoding $\Rel(T)$, used in the proof of \Cref{thm:dep(x y) lor dep(y x) L-hard}, in that it explicitly contains the index of assignments.
    This additional information is necessary for this reduction.
    Then $\MC(\dep(x, y) \lor \dep(y, x)) \colon \STRUC[\tau_t] \to \{0,1\}$ is the Boolean query that is true if and only if $T$ is a valid team and $T \models \dep(x, y) \lor \dep(y, x)$.

    Further let $\tau_{pn} = (P^2, N^2)$ be the vocabulary of structures for $\specialTwoSat$, where $P(v, c)$ encodes positive variables that occur in the same clique $c$ and $N(v, c)$ encodes negative variables.
    Then $\specialTwoSat \colon \STRUC[\tau_{pn}] \to \{0,1\}$ is the Boolean query that is true if and only if $P$ and $N$ encode a mtdf-2CNF formula $\phi$ and there is some assignment $\mathfrak{I}$ such that $\mathfrak{I} \models \phi$.

    The first-order reduction $I_{\textrm{tpn}}: \STRUC[\tau_t] \to \STRUC[\tau_{pn}]$ is as follows:
    \begin{equation*}
        I_{\textrm{tpn}} \coloneqq \lambda_{vc}\langle \texttt{true}, \psi_P, \psi_N \rangle, \qquad
        \psi_P(v, c) \coloneqq T(v, 0, c), \qquad
        \psi_N(v, c) \coloneqq T(v, 1, c),
    \end{equation*}
    where \texttt{true} means the universe is the same.
    We show that $A \in \MC(\dep(x, y) \lor \dep(y, x)) \Leftrightarrow I_{\textrm{tpn}}(A) \in \specialTwoSat$.
    Assume $A$ encodes a team $T$ such that $T = T_1 \cup T_2$, where $T_1 \models \dep(x, y)$ and $T_1 \models \dep(y, x)$.
    Then the assignment $\mathfrak{I}$ with
    \begin{itemize}
        \item $\mathfrak{I}(x_i) = 1$, if $s_i \in T_1$
        \item $\mathfrak{I}(x_i) = 0$, if $s_i \in T_2$
    \end{itemize}
    satisfies the formula encodes by $I_{\textrm{tpn}}(A)$.
    This is because, in a positive clique, at most one variable is false.
    Thus, in all clauses, at least one variable is true.
    The opposite holds for negative cliques.

    Now for the reduction in the other direction.
    Consider the following first-order reduction $I_{\textrm{pnt}}\colon \STRUC[\tau_{pn}] \to \STRUC[\tau_t]$ given by:
    \begin{equation*}
        I_{\textrm{pnt}} \coloneqq\  \lambda_{s x a}\langle \texttt{true}, \psi_T \rangle, \qquad
        \psi_T(s, x, a) \coloneqq\ (x = 0 \to N(s, a)) \land (x = 1 \to P(s, a)).
    \end{equation*}
    Notice that $I_{\textrm{pnt}}$ is simply the inverse of $I_{\textrm{tpn}}$, i.\,e., $I_{\textrm{pnt}} = I^{-1}_{\textrm{tpn}}$.
\end{proof}

\begin{figure}
    \centering
    \begin{minipage}{0.16\textwidth}
    \[
        \begin{array}{ccc}\toprule
           T & x & y \\\midrule
           s_0 & 0 & 0 \\
           s_1 & 0 & 1 \\
           s_2 & 0 & 2 \\
           s_3 & 1 & 1 \\
           s_4 & 2 & 1 \\
           s_5 & 2 & 2 \\\bottomrule
        \end{array}
    \]        
    \end{minipage}
    $\leftrightsquigarrow$
    \begin{minipage}{0.22\textwidth}
    \centering
        $T(0,0,0)$, $T(0,1,0)$, \\
        $T(1,0,0)$, $T(1,1,1)$, \\
        $T(2,0,0)$, $T(2,1,2)$, \\
        $T(3,0,1)$, $T(3,1,1)$, \\
        $T(4,0,2)$, $T(4,1,1)$, \\
        $T(5,0,2)$, $T(5,1,2)$\phantom, \\
    \end{minipage}
    $\leftrightsquigarrow$
    \begin{minipage}{0.18\textwidth}
    \centering
        $P(0,0)$, $N(0,0)$, \\
        $P(1,0)$, $N(1,1)$, \\
        $P(2,0)$, $N(2,2)$, \\
        $P(3,1)$, $N(3,1)$, \\
        $P(4,2)$, $N(4,1)$, \\
        $P(5,2)$, $N(5,2)$\phantom,
    \end{minipage}
    $\leftrightsquigarrow$
    \begin{minipage}{0.28\textwidth}
    \centering
        $(x_0 \lor x_1)$, $(x_1 \lor x_0)$, \\
        $(x_1 \lor x_2)$, $(x_2 \lor x_1)$, \\
        $(x_0 \lor x_2)$, $(x_2 \lor x_0)$, \\
        $(x_4 \lor x_5)$, $(x_5 \lor x_4)$, \\
        $(\lnot x_1 \lor \lnot x_3)$, $(\lnot x_3 \lor \lnot x_1)$, \\
        $(\lnot x_3 \lor \lnot x_4)$, $(\lnot x_4 \lor \lnot x_3)$, \\
        $(\lnot x_1 \lor \lnot x_4)$, $(\lnot x_4 \lor \lnot x_1)$, \\
        $(\lnot x_2 \lor \lnot x_5)$, $(\lnot x_5 \lor \lnot x_2)$\phantom,
    \end{minipage}
    
    \caption{Example reduction for \Cref{lem:dep(x y) lor dep(y x) equivfo specialTwoSat}.
    Team $T$ gets encodes via structure $\langle \{0,1,2,3,4,5\}, T^3 \rangle$ that gets mapped to the structure $\langle \{0,1,2,3,4,5\}, P^2, N^2 \rangle$ which encodes a mtdf-2CNF formula $\phi$; and vice versa.
    Notice how the split $\{s_1, s_3, s_5\} \models \dep(x, y)$ and $\{s_0, s_2, s_4\} \models \dep(y, x)$ corresponds to the assignment $\mathfrak{I}(x_1) = \mathfrak{I}(x_3) = \mathfrak{I}(x_5) = 1$ and $\mathfrak{I}(x_0) = \mathfrak{I}(x_2) = \mathfrak{I}(x_4) = 0$.}
    \label{fig:fo reduction}
\end{figure}

\Cref{lem:dep(x y) lor dep(y x) equivfo specialTwoSat} together with \Cref{thm:dep(x y) lor dep(y x) in L} and \Cref{thm:dep(x y) lor dep(y x) L-hard} yield the following result.
\begin{theorem}
    Satisfiability of monotone transitive dual-free 2CNF formula is $\LOGSPACE$-complete under first-order reductions.
\end{theorem}

\section{Coherence}\label{sec:coherence}
This section is devoted to identifying the coherent disjunctions of two unary dependence atoms and determining their coherence-level. 

\begin{figure}
    \centering
    $
	\begin{array}{c}
	\dep(x, y) \lor \dep(x, z)\\
	\begin{array}{ccccc}
        \toprule
            T & x & y & z \\\midrule
            s_1 & 1 & 1 & 1 \\
            s_2 & 1 & 1 & 2 \\
            s_3 & 1 & 2 & 1 \\
            s_4 & 1 & 2 & 2\\\bottomrule
            \multicolumn{4}{c}{(a)}
        \end{array}	
	\end{array}
	\quad
		\begin{array}{c}
		\dep(x) \lor \dep(y, z)\\
        \begin{array}{cccc}
        \toprule
            T & x & y & z \\\midrule
            s_1 & 1 & 1 & 1 \\
            s_2 & 1 & 1 & 2 \\
            s_3 & 2 & 2 & 1 \\
            s_4 & 2 & 2 & 2 \\\bottomrule
            \multicolumn{4}{c}{(b)}
        \end{array}
        \end{array}
        \quad
        \begin{array}{c}
        \dep(x) \lor \dep(x, y)\\
        \begin{array}{cccc}\toprule
            T & x & y \\\midrule
            s_1 & 1 & 1 \\
            s_2 & 1 & 2 \\
            s_3 & 2 & 1 \\
            s_4 & 2 & 2 \\\bottomrule
            \multicolumn{4}{c}{(c)}
        \end{array}
        \end{array}
        \quad
        \begin{array}{c}
        \dep(y) \lor \dep(x, y)\\
        \begin{array}{cccc}\toprule
            T & x & y \\\midrule
            s_1 & 1 & 1 \\
            s_2 & 2 & 1 \\
            s_3 & 3 & 2 \\
            s_4 & 4 & 2 \\\bottomrule
            \multicolumn{4}{c}{(d)}
        \end{array}
        \end{array}
	$
    \caption{Counterexamples of 3-coherence for formulas used in the proofs of Thms.~\ref{thm:dep(x y) lor dep(x z) 4-coherent} and \ref{thm:dep(x) lor dep(y z) 4-coherent}.}
    \label{fig:counterexample for constancy and dependence atom}
\end{figure}

\begin{theorem}\label{thm:dep(x y) lor dep(x z) 4-coherent}
   The formula $\dep(x, y) \lor \dep(x, z)$ has coherence-level $4$.
\end{theorem}
\begin{proof}
    We first show that $\dep(x, y) \lor \dep(x, z)$ is not 3-coherent.
    Consider the team $T$ in \Cref{fig:counterexample for constancy and dependence atom} ($a$).
    This team does not satisfy $\dep(x, y) \lor \dep(x, z)$.
    To see this, consider the two maximal subsets, $T_1 = \{s_1, s_2\}$ and $T_2 = \{s_3, s_4\}$, that satisfy $\dep(x, y)$.
    We cannot split $T = T_1 \cup T_2$, because $T_1 \not\models \dep(x, z)$ and $T_2 \not\models \dep(x, z)$.

    Every 3-element subset of $T$ satisfies $\dep(x, y) \lor \dep(x, z)$, because either $T_1$ or $T_2$ is now a singleton, thus trivially satisfying the dependence atom $\dep(x, z)$. 
    Therefore $\dep(x, y) \lor \dep(x, z)$ is not 3-coherent.
    
    Next, we show that if $T$ does not satisfies $\dep(x, y) \lor \dep(x, z)$, then there is a sub-team $T' \subseteq T$ of size $|T'| \leq 4$ that does not satisfies $\dep(x, y) \lor \dep(x, z)$ either.
    Let $T$ be a team with $T \not\models \dep(x, y) \lor \dep(x, z)$.
    If $|T| \leq 4$, then $T' = T$.
    Otherwise assume without loss of generality that all assignments in $T$ agree on $x$.
    Two assignments with different values for $x$ trivially satisfy both atoms.
    Thus, we can make this assumption by splitting $T$ into sub-teams that agree on $x$, and then consider each sub-team separately.
    
    Now, let $T' = \{s, u, v, t\}\subseteq T$, such that the following constraints hold:
    \begin{itemize}
        \item 
        $s(y) \neq t(y)$ and $s(z) \neq t(z)$; 
        \item 
        $u(y) \neq s(y)$ and $u(z) \neq t(z)$; 
        \item 
        $v(z) \neq s(z)$ and $v(y) \neq t(y)$.
    \end{itemize}
    We first show that $T' \not\models \dep(x, y) \lor \dep(x, z)$ is true and then that such a $T' \subseteq T$ always exists (see \Cref{fig:team stuv}).
    By the constraints on $y$, we have that $\{s, v\}$ and $\{u, t\}$ are the two maximal sub-teams that satisfy $\dep(x, y)$.
    Therefore at least one of the two sub-teams must also satisfy $\dep(x, z)$ for the disjunction to be true.
    But $\{s, v\} \not\models \dep(x, z)$, because $s(x) = v(x)$ and $s(z) \neq v(z)$; analogously $\{u, t\} \not\models \dep(x, z)$.
    Thus $T' \not\models \dep(x, y) \lor \dep(x, z)$ is true.

    Let $s \in T$ and assume no $t \in T$ meets the constraints above.
    Then $T$ can be split via $T_1 = \{w \in T \mid w(y) = s(y)\} \models \dep(x, y)$ and $T_2 = \{w \in T \mid w(z) = s(z)\} \models \dep(x, z)$.
    Next, assume there are $s, t$ and $u$, but no $v$.
    Then we split $T$ via $T_1 = \{w \in T \mid w(y) = t(y)\} \models \dep(x, y)$ and $T_2 = \{w \in T \mid w(z) = s(z)\} \models \dep(x, z)$.
    The last case is analogous.
    This is a contradiction to $T \not\models \dep(x, y) \lor \dep(x, z)$, hence such a sub-team must always exists.
    \begin{figure}
        \centering
        $
        \begin{array}{cccc}\toprule
            T & x & y & z \\\midrule
            s_1 & 0 & 0 & 0 \\
            s_2 & 0 & 0 & 1 \\
            s_3 & 0 & 1 & 2 \\
            s_4 & 0 & 2 & 0 \\
            s_5 & 0 & 2 & 1 \\\bottomrule
        \end{array}
        $
        $\leadsto$
        \begin{tikzpicture}[scale=2.25, baseline=(current bounding box.center)]
            \node (s) at (0,0) {$s_1$};
            \node (t) at (0,-1) {$s_5$};
            \node (u) at (-0.5,-0.5) {$s_3$};
            \node (v) at ( 0.5,-0.5) {$s_3$};
    
            \draw (s) -- node[right] {$yz$} (t);
            \draw (s) -- node[left] {$y$} (u);
            \draw (s) -- node[right] {$z$} (v);
            \draw (u) -- node[left] {$z$} (t);
            \draw (v) -- node[right] {$y$} (t);
        \end{tikzpicture}
        \qquad
        $
        \begin{array}{cccc}\toprule
            T & x & y & z \\\midrule
            s_1 & 0 & 0 & 0 \\
            s_2 & 0 & 0 & 2 \\
            s_3 & 0 & 1 & 1 \\
            s_4 & 0 & 2 & 1 \\
            s_5 & 1 & 2 & 0 \\\bottomrule
        \end{array}
        $
        $\leadsto$
        \begin{tikzpicture}[scale=2.25, baseline=(current bounding box.center)]
            \node (s) at (0,0) {$s_2$};
            \node (t) at (0,-1) {$s_3$};
            \node (u) at (-0.5,-0.5) {$-$};
            \node (v) at ( 0.5,-0.5) {$s_1$};
    
            \draw (s) -- node[right] {$yz$} (t);
            \draw (s) -- node[left] {$y$} (u);
            \draw (s) -- node[right] {$z$} (v);
            \draw (u) -- node[left] {$z$} (t);
            \draw (v) -- node[right] {$y$} (t);
        \end{tikzpicture}
        \caption{Two examples for the construction in \Cref{thm:dep(x y) lor dep(x z) 4-coherent}.
        Edge label denote the variable values where the assignments from the vertices differ, e.g., edge $\{s_1,s_3\}$ with label $y$ means that $s_1(y)\neq s_3(y)$.
        On the left we have that $s_1$ corresponds to $s$, $s_3$ to $u$ and $v$, and $s_5$ to $t$ from the construction.
        }
        \label{fig:team stuv}
    \end{figure}
\end{proof}
As a special case of the previous lemma, we deduce the following corollary.
\begin{corollary}\label{cor:dep(y) lor dep(z) is 4-coherent}
    The formula $\dep(y) \lor \dep(z)$ has coherence-level $4$.
\end{corollary}
\begin{restatable}{theorem}{coherenceOfMix}\label{thm:dep(x) lor dep(y z) 4-coherent}
    Each of the following formulas has coherence-level $4$:
    
    \noindent \itemstyle{(1.)} $\dep(x) \lor \dep(y, z)$;\quad
    \itemstyle{(2.)} $\dep(x) \lor \dep(x, y)$;\quad
    \itemstyle{(3.)} $\dep(y) \lor \dep(x, y)$.
\end{restatable}
\begin{proof}
    We first demonstrate that $\dep(x) \lor \dep(y, z)$, $\dep(x) \lor \dep(x, y)$ and $\dep(y) \lor \dep(x, y)$ are not 3-coherent.
    
\itemstyle{(1.)}
    Consider the team $T$ in \Cref{fig:counterexample for constancy and dependence atom} ($b$).
    We have that $T \not\models \dep(x) \lor \dep(y, z)$ which is straightforward to see; consider for example the splits $\{s_1,s_2\} \models \dep(x)$ or $\{s_3,s_4\} \models \dep(x)$ and notice that in both cases $\dep(y,z)$ is falsified by $\{s_3,s_4$\} and $\{s_1,s_2\}$ respectively.
    
    Now, all 3-element sub-teams have a split $T' = T_1 \cup T_2$, where $T_1$ is of size two and $T_1 \models \dep(x)$ and $T_2$ is of size one and therefore satisfies any dependence atom, in particular $T_2 \models \dep(y, z)$.

\itemstyle{(2.)}
    Not $3$-coherent, because of the team in \Cref{fig:counterexample for constancy and dependence atom} ($c$).

\itemstyle{(3.)}
    Not $3$-coherent, because of the team in \Cref{fig:counterexample for constancy and dependence atom} ($d$). 

    For coherence it suffices to only consider $\dep(x) \lor \dep(y, z)$.
    The other two cases directly follow from the first.

    We show that if $T \not\models \dep(x) \lor \dep(y, z)$, then there must be a sub-team $T' \subseteq T$ with $|T'| \leq 4$ and $T' \not\models \dep(x) \lor \dep(y, z)$.
    Start with a pair of assignments $s_1, s_2 \in T$ for which $\{s_1,s_2\} \not\models \dep(y,z)$ holds, i.\,e., $s_1(y) = s_2(y)$ and $s_2(z) \neq s_2(z)$.
    Such a pair has to exist in $T$ to falsify the disjunction.
    Now, there are two cases for the assignment of $x$.% in $s_1$ and $s_2$.
    \begin{description}
        \item[Case 1:] 
        The assignments $s_1$ and $s_2$ have the same value for $x$, i.\,e., $s_1(x) = s_2(x)$.
        Then for $T$ to falsify the disjunction, there must be a second pair $\{s_3, s_4\} \not\models \dep(y, z)$ with $s_3(x) \neq s_1(x)$ and $s_4(x) \neq s_1(x)$.
        but then $T' = \{s_1, s_2, s_3, s_4\}$ and $T' \not\models \dep(x) \lor \dep(y, z)$.
        To see this, assume there is a split $T_1 \cup T_2 = \{s_1, s_2, s_3, s_4\}$ with $T_1 \models \dep(x)$ and $T_2 \models \dep(y,z)$.
        If $s_1 \in T_1$, then $s_3 \in T_2$, but $s_4$ cannot be in $T_1$ or $T_2$.
        If $s_1 \in T_2$, then $s_2 \in T_1$, therefore $s_3$ and $s_4$ must be in $T_2$ which is not possible.
        Therefore no split exists.
    \item[Case 2:]
        The assignments $s_1$ and $s_2$ have different values for $x$, i.\,.e, $s_1(x) \neq s_2(x)$.
        Then the pair satisfies neither $\dep(x)$ nor $\dep(y,z)$, so in any potential split $s_1 \in T_1$ implies $s_2 \in T_2$ and vice versa.
        Since the original team falsified the disjunction, there must be an assignment $s_3 \in T$ that can neither be in $T_1$ nor $T_2$ if $s_1 \in T_1$ and $s_2 \in T_2$.
        In the same way there must be an assignment $s_4 \in T$ that is not in $T_1$ or $T_2$ for $s_1 \in T_2$ and $s_2 \in T_1$.
        Therefore $T' = \{s_1,s_2,s_3,s_4\}$ also has no split satisfying $\dep(x) \lor \dep(y,z)$. \qedhere
    \end{description}
\end{proof}

Finally, we turn towards the three incoherent cases. 
For these formulas, there exists no $k\in\mathbb N$, such that the respective formula is $k$-coherent. 
Notice that the incoherence result can here be deduced from the fact that $\FO\subsetneq\LOGSPACE$. 
Nevertheless, we provide independent proofs for these results.

\begin{restatable}{theorem}{incoherentTheorem}\label{thm:dep(x z) lor dep(y z) incoherent}\label{thm:dep(x y) lor dep(y x) incoherent}\label{cor:dep(x y) lor dep(y z) incoherent}
The following three formulas are incoherent:

\noindent \itemstyle{(1.)} $\dep(x, y) \lor \dep(y, x)$;\quad
\itemstyle{(2.)} $\dep(x, y) \lor \dep(y, z)$;\quad
\itemstyle{(3.)} $\dep(x, z) \lor \dep(y, z)$.
\end{restatable}
\begin{proof}
\itemstyle{(1.)}
    We show that for all even $n \in \mathbb{N}$, there is a team $T$ of size $|T| = n+1$ with $T \not\models \dep(x, y) \lor \dep(y, x)$, but for all $n$-element sub-teams $T' \models \dep(x, y) \lor \dep(y, x)$ holds.

    \begin{figure}
        \centering
        $
        \begin{array}{ccc}\toprule
            T_n & x & y \\ \midrule
            s_0 & 1 & \frac{n}{2} \\ 
            s_1 & 1 & 1 \\ 
            s_2 & 1 & 2 \\ 
            s_3 & 2 & 2 \\ 
            s_4 & 2 & 3 \\ 
            \vdots & \vdots & \vdots \\
            s_{n-1} & \frac{n}{2} & \frac{n}{2} \\ 
            s_{n\phantom{+0}} & \frac{n}{2} & 1 \\\bottomrule         
        \end{array}
        \qquad
        \begin{array}{ccc}\toprule
            T' & x & y \\ \midrule
            s_0 & 1 & \frac{n}{2} \\ 
            s_3 & 2 & 2 \\ 
            \vdots & \vdots & \vdots \\
            s_{n\phantom{+0}} & \frac{n}{2} & 1 \\\hdashline
            s_1 & 1 & 1 \\ 
            s_2 & 1 & 2 \\ 
            \vdots & \vdots & \vdots \\
            s_{n-1} & \frac{n}{2} & \frac{n}{2} \\\bottomrule  
        \end{array}
        \qquad\qquad
        \begin{array}{cccc}\toprule
            T_n & x & y & z \\ \midrule
            s_0 & 1 & \frac{n}{2} & 3 \\ 
            s_1 & 1 & 1 & 1 \\ 
            s_2 & 1 & 2 & 2 \\ 
            s_3 & 2 & 2 & 1 \\ 
            s_4 & 2 & 3 & 2 \\ 
            \vdots & \vdots & \vdots & \vdots \\
            s_{n-1} & \frac{n}{2} & \frac{n}{2} & 1 \\ 
            s_{n\phantom{+0}} & \frac{n}{2} & 1 & 2 \\\bottomrule          
        \end{array}
        \qquad
        \begin{array}{cccc}\toprule
            T' & x & y & z \\ \midrule
            s_0 & 1 & \frac{n}{2} & 3 \\ 
            s_3 & 2 & 2 & 1 \\ 
            \vdots & \vdots & \vdots & \vdots \\
            s_{n\phantom{+0}} & \frac{n}{2} & 1 & 2 \\\hdashline
            s_1 & 1 & 1 & 1 \\ 
            s_2 & 1 & 2 & 2 \\ 
            \vdots & \vdots & \vdots & \vdots \\
            s_{n-1} & \frac{n}{2} & \frac{n}{2} & 1 \\ \bottomrule
        \end{array}
        $\\$\,$
        \caption{Team $T_n$ in the proof of \Cref{thm:dep(x y) lor dep(y x) incoherent}, and constructed split of an example sub-team $T' = T_n \setminus\{s_{4}\}$.
        On the left for $\dep(x, y) \lor \dep(y, x)$ and on the right for $\dep(x, z) \lor \dep(y, z)$.}
        \label{tab:incoherent}        
    \end{figure}

    Let the team $T_n$ be as depicted in \Cref{tab:incoherent}.
    We show that it is impossible to split $T_n = T_1 \cup T_2$ such that $T_1 \models \dep(x, y)$ and $T_2 \models \dep(y, x)$.
    Start with assignment $s_1$ and choose $s_1 \in T_1$.
    Next, $s_2 \in T_2$ is the only choice, because $T_1 = \{s_1, s_2\} \not\models \dep(x, y)$. 
    Then $s_3 \in T_1$, because $T_2 = \{s_2, s_3\} \not\models \dep(y, x)$, but $T_1 = \{s_1, s_3\} \models \dep(x, y)$.
    Continue this procedure until all assignments but $s_0$ are either in $T_1$ or $T_2$.
    Now, $s_0$ cannot be in $T_1$, because $\{s_0, s_1\} \not\models \dep(x, y)$ and it cannot be in $T_2$, because $\{s_0, s_{n-2}\} \not\models \dep(y, x)$.
    If instead we chose $s_1 \in T_2$ and repeated the process backwards, that is continued with $s_n$, then $\{s_0, s_2\} \not \models \dep(x, y)$ and $\{s_0, s_{n-1}\} \not \models \dep(y, x)$.
    Since $s_1$ has to be in either $T_1$ or $T_2$, we can conclude that no split exists and $T_n \not\models \dep(x, y) \lor \dep(y, x)$.

    Next, we show that all sub-teams of size $n$ satisfy the disjunction.
    The sub-team $T_n \setminus \{s_0\} \models \dep(x, y) \lor \dep(y, x)$ given the split described above.
    Otherwise, let $T' = T_n \setminus \{s_i\}$ and consider the split above where $s_1 \in T_1$.
    We define a new split for $T'$ as follows
    \begin{align*}
        T'_1 &= \{s_j \in T_1 \mid 1 < j < i\} \cup \{s_j \in T_2 \mid j > i\} \cup \{s_0\}\\
        T'_2 &= \{s_j \in T_2 \mid 1 < j < i\} \cup \{s_j \in T_1 \mid j > i\} \cup \{s_1\}   
    \end{align*}
    Now, $T'_1 \models \dep(x, y)$ and $T'_2 \models \dep(y, x)$.

    We will present the case that $s_i \in T_1$; the case $s_i \in T_2$ follows analogously.
    If $s_i \in T_1$, then $s_{i+1}$ has to be in $T_2$ as shown above.
    Now, since $s_i \not\in T'$, it is possible that $s_{i+1} \in T'_1$.
    From this $s_{i+2} \in T'_2, s_{i+3} \in T'_1, \dots$ follows immediately.
    This chain leads to $s_n \in T'_1$, which allows $s_1 \in T'_2$, which in turn makes $s_0 \in T'_1$ possible.
    Therefore we have that $T'_1 \models \dep(x, y)$ and $T'_2 \models \dep(y, x)$.

    We have shown that for every even $n \in \mathbb{N}$ we can construct a team such that $\dep(x, y) \lor \dep(y, x)$ is not $n$-coherent, that is $\dep(x, y) \lor \dep(y, x)$ is incoherent.

\itemstyle{(2.)}
    Extend the team $T_n$ in \itemstyle{(1.)} with variable $z$, such that $s(z) = s(x)$ for all $s \in T_n$.
    Then incoherence of $\dep(x,y) \lor \dep(y, x)$ follows immediately from the incoherence of $\dep(x,y) \lor \dep(y,x)$.

\itemstyle{(3.)}
    For $\dep(x, z) \lor \dep(y, z)$ extend $T_n$ with variable $z$ such that if $s(x) = s'(x)$ or $s(y) = s'(y)$, then $s(z) \neq s'(z)$ for all $s, s' \in T$.
    For example see team $T_n$ in \Cref{tab:incoherent}.
    It is easy to check, that the arguments of \itemstyle{(1.)} also hold true in such a team for $\dep(x, z) \lor \dep(y, z)$.
\end{proof}
By combining the results from \Cref{sec:mc} and \ref{sec:coherence}, we obtain the following classification for the complexity of the model checking of disjunctions of two dependence or constancy atoms.
\begin{theorem}
    The following statements are true for the model checking problem of disjunctions of two unary dependence or constancy atoms:
    \begin{enumerate}[(1.)]
        \item If $\varphi$ is one of the formulas $\dep(x,y)\lor\dep(z,u)$, $\dep(x,z)\lor\dep(y,z)$, $\dep(x,y)\lor\dep(y,z)$, then the problem $\MC(\varphi)$ is $\NL$-complete.
        \item The problem $\MC(\dep(x,y)\lor\dep(y,x))$ is $\LOGSPACE$-complete.
        \item In all other cases, the problem $\MC(\varphi)$ is in $\FO$, hence  also in uniform $\ACzero$.
    \end{enumerate}
\end{theorem}

\section{Disjunctions of higher-arity dependence atoms}
In this section, we consider disjunctions of two dependence atoms whose arity  may be higher than one.  We leverage the preceding results about disjunctions of two unary dependence atoms to classify the complexity of disjunctions of two
dependence atoms of higher arity.

Before stating and proving the main result of this section, we start with two lemmas.
\begin{lemma}\label{lem:simulate higher arity}
    Consider a $\DL$-formula of the form $\dep(\bar x, y)\lor \dep(\bar u, v)$, where $\bar x = x_1 \cdots x_n$ and $\bar u = u_1 \cdots u_m$.
    For every team $T$ with domain $\{x_1, \dots, x_n, y, u_1, \dots, u_m, v\}$ and range $A$, there is  a team $T'$ with domain $\{x',y, u_1,\dots,u_m,v\}$ and range $A^n \cup A$, such that  
    \[
        T \models \dep(\bar x, y)\lor \dep(\bar u, v) \quad \mbox{if and only if} \quad  T' \models \dep(x', y)\lor \dep(\bar u, v).
    \]
\end{lemma}
\begin{proof} 
Assume $T = \{s_1, \dots, s_k\}$ and consider the assignments $s'_i$, $1\leq i\leq k$, with domain $\{x',y,u_1,\ldots,u_m,y\}$ and range $A^n\cup A$ defined as follows: $s'_i(x') = (s_i(x_1),\ldots,s_i(x_n)) = s_i(\bar x)$ and $s'_i(w)=s_i(w)$, if $w$ is one of the variables $y,u_1,\ldots,u_m,v$.
Furthermore, consider the team $T'= \{s'_1, \dots, s'_k\}$, with domain $\{x',y,u_1,\dots,u_m,v\}$ and range $A^n\cup A$.
    
Suppose $T\models\dep(\bar x, y)\lor \dep(\bar u, v)$ and let $T = T_1\cup T_2$ be a split, such that $T_1\models\dep(\bar x,y)$ and $T_2\models \dep(\bar u,v)$. 
Define $T_1'= \{s_i'\in T'\mid s_i\in T_1\} $ and $T_2'= \{s_i'\in T'\mid s_i\in T_2\}$. 
Now, the following hold:
    \begin{align*}
        T_1 \models \dep(\bar x, y) 
        & ~\mbox{iff}~ \forall s, t \in T_1\colon s(\bar x) = t(\bar x) \Rightarrow s(y) = t(y) \\
        % & \mbox{iff} \forall s, t \in T_1\colon (s(x_1), \dots, s(x_n)) = (t(x_1), \dots, t(x_n)) \Rightarrow s(y) {=} t(y) \\
        & ~\mbox{iff}~ \forall s', t' \in T_1'\colon s'(x') = t'(x') \Rightarrow s'(y) = t'(y) \\
        & ~\mbox{iff}~ T_1' \models \dep(x', y).
    \end{align*}
Furthermore,  $T_2'\models \dep(\bar u,v)$, since $T_2\upharpoonright\{u_1,\dots,u_m,v\}=T'_2\upharpoonright\{u_1,\dots,u_m,v\}$.
Thus, if $T \models \dep(\bar x, y)\lor \dep(\bar u, v)$, then $T' \models \dep(x', y)\lor \dep(\bar u, v)$.
The converse is established using a similar argument.
\end{proof}

By applying Lemma \ref{lem:simulate higher arity} twice, we obtain the following result.

\begin{corollary} \label{cor:simulate higher arity}
% Consider a $\DL$-formula of the form $\dep(x_1, \dots, x_n, y)\lor \dep(u_1, \dots, u_m, v)$. 
Consider a $\DL$-formula of the form $\dep(\bar x, y)\lor \dep(\bar u, v)$, where $\bar x = x_1 \cdots x_n$ and $\bar u = u_1 \cdots, u_m$.
For every team $T$ with domain $\{x_1, \dots, x_n, y, u_1, \dots, u_m, v\}$ and range $A$, there is  a team $T'$ with domain $\{x',y, u',v\}$ and range $A^n \cup A^m \cup A$, such that  
    \[
        T \models \dep(\bar x, y)\lor \dep(\bar u, v) \quad \mbox{if and only if} \quad T' \models \dep(x', y)\lor \dep(u', v).
    \]
\end{corollary}

The next lemma establishes the coherence of  one of the simplest higher-arity disjunctions.%
\begin{lemma} \label{lem:coherence-4}
    The formula $\dep(x, z) \lor \dep(xy, u)$ has coherence-level 4.
\end{lemma}
\begin{proof}
    To show that   $\dep(x, z) \lor \dep(xy, u)$ is not 3-coherent, notice that this formula
    is equivalent to $\dep(x, z) \lor \dep(x, u)$ on  every team $T$ that is constant on $y$, i.e., $s(y) = t(y)$ holds, for all $s, t \in T$.
    Thus, the supplement team $T_f^{y}$ given by $T$ in \Cref{fig:counterexample for constancy and dependence atom} (a) (and with suitable variable renamings) and constant function $f(s) = a$ is a counterexample to the 3-coherence of $\dep(x, z) \lor \dep(xy, u)$.

    To show that   $\dep(x, z) \lor \dep(xy, u)$ is 4-coherent, consider an arbitrary team $T$ with domain $\{x,z,y,u\}$. 
    Without loss of generality, we may  assume that $T$ is constant on $x$.
    This is so because $T$ can be split into sub-teams $T_{a_i} = \{s \in T \mid s(x) = a_i \in A\}$ that are constant on $x$. 
    If $T$ falsifies $\dep(x, z) \lor \dep(xy, u)$, then at least one of these sub-teams also falsifies $\dep(x, z) \lor \dep(xy, u)$.
    Indeed, if every sub-team $T_{a_i}$ satisfied $\dep(x, z) \lor \dep(xy, u)$, then their union $T$ would also satisfy $\dep(x, z) \lor \dep(xy, u)$, since every 
    team $T'$ consisting of  two assignments  $s \in T_{a_j}$ and   $t \in T_{a_k}$ with $a_j\not =  a_k$ satisfies both dependence atoms (because $s(x) = a_j \neq a_k = t(x)$).
    
    Now, since $T$ is constant on $x$, whether or not 
    $T\models  \dep(x, z)$  depends only on the values of $z$ and so  $\dep(x, z)$ acts like the constancy atom $\dep(z)$.
    Similarly, whether or not $T\models \dep(xy, u)$,  depends only on the values of $y$ and $u$, and so $\dep(xy,u)$ acts like  the unary dependence atom  into $\dep(y,u)$. 
    Therefore, we have that $T \models \dep(x, z) \lor \dep(xy, u)$ if and only if $T \models \dep(z) \lor \dep(y, u)$.
    As the latter disjunction was shown to be 4-coherent in \Cref{thm:dep(x) lor dep(y z) 4-coherent}, we have that $\dep(x, z) \lor \dep(xy, u)$ is  4-coherent as well.
\end{proof}

While \Cref{lem:simulate higher arity} and \Cref{cor:simulate higher arity} allow us to go from higher-arity atoms to unary atoms, our next result concerns the other direction, i.\,e., from unary atoms to higher-arity atoms.
% Next, we proof the converse direction of \Cref{lem:simulate higher arity} and \Cref{cor:simulate higher arity} to proof lower-bounds for model checking problems of higher arity atoms.
\begin{lemma}
    Let $T$ be a team over range $A$ and let $f\colon T \to A$ be the constant function $f = 0$.
    We have that $T \models \dep(\bar x, z)$ if and only if $T^{\bar w}_f \models \dep(\bar x \bar w, z)$. 
\end{lemma}
\begin{proof}
    Assume $T \models \dep(\bar x, z)$.
    Then for all $t, t' \in T$, we have that if $t(\bar x) = t'(\bar x)$, then $t(z) = t'(z)$.
    Since $\bar w$ is constant across all assignments in $T^{\bar w}_f$, we also have that $t(\bar x\bar w) = t'(\bar x\bar w)$ implies $t(z) = t'(z)$, for all $t, t' \in T^{\bar w}_f$.
    Thus, $T^{\bar w}_f \models \dep(\bar x \bar w, z)$.

    For the other direction, assume $T \not\models \dep(\bar x, z)$.
    Then there are two assignments $t, t' \in T$, such that $t(\bar x) = t'(\bar x)$, but $t(z) \neq t'(z)$.
    Clearly, for the same two assignments in $T^{\bar w}_f$, we have that $t(\bar x\bar w) = t'(\bar x\bar w)$ and $t(z) \neq t'(z)$.
    Consequently, $T^{\bar w}_f \not\models \dep(\bar x \bar w, z)$.
\end{proof} 

\begin{corollary}\label{cor:hardness to higher arity}
    Let $T$ be a team over range $A$, and $f\colon T \to A$ be the constant function $f = 0$.
    Suppose $\bar w$ and $\bar v$ are variables that do not occur in $\bar x$ or $\bar y$ and are different from $z$ and $u$.
    Then $T \models \dep(\bar x, z) \lor \dep(\bar y, u)$ if and only if $T^{\bar w\bar v}_f \models \dep(\bar x\bar w, z) \lor \dep(\bar y\bar v, u)$. 
\end{corollary}

We need to consider one more special case, before the full classification can be given.
\begin{lemma}\label{lem:NL higher base cases}
    % The following statements are true:
    % \begin{enumerate}
        % \item The problem $\MC(\dep(wu, z) \lor \dep(yu, z))$ is $\NL$-complete.
        % \item 
        The problem $\MC(\dep(wu, z) \lor \dep(z, u))$ is $\NL$-complete.
    % \end{enumerate}
\end{lemma}
\begin{proof}
    Membership follows from \Cref{prop:lor of 2 coherent in NL}.
    % For $\NL$-hardness of \itemstyle{(1.)}, we reduce from $\MC(\dep(w, z) \lor \dep(y, z))$.
    % Consider the simple reduction $T \mapsto S$, where $S$ is the supplement team $T^u_f$ and $f = 0$ is a constant function.
    % Let $T = T_1 \cup T_2$ be a split and $S_1 = (T_1)^u_f$ and $S_2 = (T_2)^u_f$. 
    % We now have 
    % \begin{align*}
    %     &T_1 \models \dep(w, z) \\
    %     &\text{iff } \forall t, t' \in T_1\colon t(w) = t'(w) \Rightarrow t(z) = t'(z) \\
    %     &\text{iff } \forall s, s' \in S_1\colon s(w) = s'(w) \text{ and } s(u) = s'(u) \Rightarrow s(z) = s'(z) \tag{$\star$}\\
    %     &\text{iff } S_1 \models \dep(wu, z).
    % \end{align*}
    % Equivalence $(\star)$ holds, because $s(u) = 0 = s'(u)$ for all $s, s' \in S_1$.
    % Likewise $T_2 \models \dep(y, z)$ if and only if $S_2 \models \dep(yu, z)$.
    % Consequently, $T \models \dep(w, z) \lor \dep(y, z)$ if and only if $S \models \dep(wu, z) \lor \dep(yu, z)$.
    % Next, $\NL$-hardness of \itemstyle{(2.)}.
    For $\NL$-hardness, we reduce from monotone $\TwoSAT$. 
    Let $\theta$ be a propositional formula of the form $\bigwedge_{i \in I} C_i$, where every clause is either positive or negative.
    For each clause $C_i$, we construct a team $T_i$ of size two and domain $\{w, u, z\}$.
    We let $w$ range over $I$, $u$ over $\{0, 1\}$ and $z$ over the propositions in $\theta$. 
    Now, each proposition $p_k$ in $C_i$ gives rise to one assignment $t$ in $T_i$, such that $t(w) = i$, $t(u) = 1$ if $C_i$ is a positive clause and $0$ otherwise, and $t(z) = p_k$.
    Finally, $T = \bigcup_{i \in I} T_i$. 
    Clearly, $T$ can be computed in logarithmic space with a binary counter for the clause index.

    Assume $F$ is a satisfying assignment of $\theta$.
    We partition $T$ as follows: $S_2 = \{s \in T \mid F(s(z)) = s(u)\}$ is the set of literals satisfied  by  $F$, and $S_1 = T \setminus S_2$.
    Clearly, $S_2 \models \dep(z, u)$, 
    because if $s$ and $s'$ are two assignments in $S_2$ such that 
    $s(z) = p = s'(z)$ for some proposition $p$, the
    $s(u)=F(p)=s'(u)$ holds ($F$ assigns a unique value to each proposition).
    %since $s(z) = p = s'(z)$ for any pair of assignments $s, s' \in S_2$ and proposition $p$ implies $s(u) = F(p) = s'(u)$---$F$ is a consistent assignment.
    To show that $S_1 \models \dep(wu, z)$, note that $F$ makes each clause   $C_i$ true, since $F$ is a satisfying assignment.
    Therefore, at most one of the two assignments in the associated team $T_i$ is in the $S_1$ part.
    Consequently, no two assignments in $S_1$ have the same value for $w$, thus $S_1 \models \dep(wu, z)$.

    For the other direction, assume $T = S_1 \cup S_2$, $S_1 \models \dep(wu, z)$ and $S_2 \models \dep(z, u)$.
    Define the function $F$ as follows: $F(p) = 1$ if there is an assignment $s \in S_2$, such that $s(z) = p$ and $s(u) = 1$; otherwise $F(p) = 0$.
    Next, we show that $F$ is a satisfying assignment for $\theta$.
    Consider the team $T_i$ associated with the  clause $C_i$.
    For the two assignments $t, t' \in T_i$, we have that $t(w) = i = t'(w)$ and, since $\theta$ is monotone, either $t(u) = 1 = t'(u)$ or $t(u) = 0 = t'(u)$.
    Furthermore, $t(z) \neq t'(z)$.
    Consequently, $T_i \not\models \dep(wu, z)$ and, as not both can be in $S_1$, one of the two assignments must be in $S_2$.
    This directly implies $F(C_i) = 1$, because $S_2$ models the truth of $\theta$.
\end{proof}

We are now ready to state and prove a classification theorem for the complexity of the model checking problem for disjunctions of dependence atoms of higher arities. Note that, when a dependence atom $\dep(\bar y,u)$ of higher arity is considered, then, by  the semantics of dependence logic $\DL$, the order $y_1,\ldots,y_n$ of the variables in $\bar y$ is immaterial. Also, recall that in a dependence atom $\dep(\bar{y},u)$, the variable $u$ and all the variables  
in $\bar{y}$ are distinct from each other. In particular, an expression of the form $\dep(uy,u)$ is not a dependence atom.

\begin{theorem} \label{thm:higher-classification}
    Let $\dep(\bar x, z)$ and $\dep(\bar y, u)$ be two dependence atoms, where some of the variables of the first atom may coincide with some of the variables of the second.
    Then the following hold:
    \begin{enumerate}[(1.)]
        \item If $\bar x \subseteq \bar y$ or $\bar y \subseteq \bar x$, then 
        \(
        \MC(\dep(\bar x,z) \lor \dep(\bar y,u))
        \)
        is in $\FO$. In other words, both  $\MC(\dep(\bar x,z) \lor \dep(\bar x \bar y',u))
        $ and $\MC(\dep(\bar x' \bar y,z) \lor \dep(\bar y,u))
        $ are in $\FO$.
        \item If $\bar x = \bar x'u$ and $\bar y = \bar x'z$, then 
        \(
        \MC(\dep(\bar x,z) \lor \dep(\bar y,u))
        \)
        is $\LOGSPACE$-complete.
In other words, $\MC(\dep(\bar{x}u,z) \vee \dep(\bar{x}z,u))$ is $\LOGSPACE$-complete.
        \item In all other cases, 
        \(
        \MC(\dep(\bar x,z) \lor \dep(\bar y,u))
        \)
        is $\NL$-complete.
    \end{enumerate}
\end{theorem}

\begin{proof} 
    For the first part, we distinguish the three cases $\bar x = \bar y$, $\bar x \subsetneq \bar y$ and $\bar y \subsetneq \bar x$. 
    Since the last two cases are symmetric, we will only consider one in detail.
    
    \noindent \emph{Case} ``$\bar x = \bar y$'':
    The formula $\dep(\bar x, z) \lor \dep(\bar x, u)$ has coherence-level 4, because it is equivalent to $\dep(x, z) \lor \dep(x, u)$, by Corollary \ref{cor:simulate higher arity}.
    Thus, with \Cref{{coherent-FO-reduction}} and \Cref{thm:dep(x y) lor dep(x z) 4-coherent}, it follows that $\MC(\dep(\bar x, z) \lor \dep(\bar x, u))$ is in $\FO$.
    
    \noindent \emph{Case} ``$\bar x \subsetneq \bar y$'':
    Let $\bar y' = \bar y \setminus \bar x$ be the variables that occur in the second, but not the first dependence atom.
    We encode $\bar x$ as a single variable $x$, and likewise $\bar y'$ as $y'$.
    Then, the disjunction $\dep(\bar x, z) \lor \dep(\bar x\bar y', u)$ is equivalent to the disjunction $\dep(x, z) \lor \dep(xy', u)$, which is 4-coherent by \Cref{lem:coherence-4}.
    Thus, $\MC(\dep(\bar x, z) \lor \dep(\bar y, u))$ is in $\FO$, by \Cref{coherent-FO-reduction}.

    For the second part, we first consider $\LOGSPACE$-hardness.
    There is a straightforward reduction from the unary case, i.\,e., the model checking of $\dep(u, z) \lor \dep(z, u)$, which was shown to be $\LOGSPACE$-complete in \Cref{sec:L-complete}.
    From \Cref{cor:hardness to higher arity}, we know that if $T \models \dep(u, z) \lor \dep(z, u)$, then $T^{\bar x'}_f \models \dep(\bar x'u, z) \lor \dep(\bar x'z, u)$, for the constant function $f = 0$.
    Adding new variables, that are constant in every assignment, can clearly be done via a first-order reduction.
    Membership in $\LOGSPACE$ only requires changing the definition of the graph $G_T$ in \Cref{def:G_T} as follows: $V_T = \Range(\bar x'u) \cup \Range(\bar x'z)$ and $E_T = \{\{s(\bar x'u), s(\bar x'z)\} \mid s \in T\}$.
    This adjustment then naturally separates assignments with different $\bar x'$ values into different connected components, because they never falsify either dependence atom, regardless of the values in $u$ and $z$.

    For the third part, we differentiate between two cases, namely $z = u$ and $z \neq u$.
    In both cases, membership in $\NL$ follows from \Cref{prop:lor of 2 coherent in NL}.
    
    \noindent \emph{Case} ``$z = u$'': 
    For $\NL$-hardness of $\MC(\dep(\bar x, z) \lor \dep(\bar y, z))$, we proceed by a reduction from the problem $\MC(\dep(x, z) \lor \dep(y, z))$, which is $\NL$-complete by \Cref{lem:NL higher base cases}.
    First, we rename the variables in $\bar x$ and $\bar y$, such that $x \in \bar x \setminus \bar y$ and $y \in \bar y \setminus \bar x$.
    Then, let $\bar v$ and $\bar v'$ be the variables, such that $\bar x = \bar v x$ and $\bar y = \bar v y$.
    The reduction now simply constructs the supplement team $T^{\bar v\bar v'}_f$, where $f$ is the constant function $f = 0$.
    Correctness, i.\,e., $T \models \dep(x, z) \lor \dep(y, z)$ if and only if $T^{\bar v\bar v'}_f \models \dep(\bar x, z) \lor \dep(\bar y, z)$, follows directly from \Cref{cor:hardness to higher arity}, which is applicable, since $z \not\in \bar x \cup \bar y$.
    
    \noindent \emph{Case} ``$z \neq u$'':
    First, if $|\bar x| = |\bar y| = 1$, then we are in the unary case where hardness was established in \Cref{prop:dep(x y) lor dep(z w) NL-hard}.
    Thus, suppose $|\bar x| \neq 1$ or $|\bar y| \neq 1$ (or both).
    We give a reduction from $\MC(\dep(wu, z) \lor \dep(z, u))$, which is $\NL$-complete by \Cref{lem:NL higher base cases}, to $\MC(\dep(\bar x, z) \lor \dep(\bar y, u))$.     
    First, notice that there is a variable $w$, such that either $w \in \bar x$ and $w \not\in \bar y \cup \{u\}$ or $w \in \bar y$ and $w \not\in \bar x \cup \{z\}$.
    With out loss of generality assume $w \in \bar x$.
    Then, we again rename the variables, such that $\bar x = \bar vwu$ and $\bar y = \bar v'z$.
    Next, given a team $T$, construct the team $T^{\bar v\bar v'}_f$, with $f$ the constant function $f = 0$.
    Correctness, once more, follows from \Cref{cor:hardness to higher arity}.
\end{proof}

Observe that, 
whenever $\MC(\varphi)$ was shown to be  in $\FO$,  the proof of Theorem 
\ref{thm:higher-classification}  actually showed that $\varphi$ is coherent.
Combined with the earlier results about unary dependence atoms, we obtain the following characterization of coherence of disjunctions of dependence atoms.

\begin{corollary}
    Let $\varphi$ be a disjunction of two dependence atoms of arbitrary arities. Then $\varphi$ is coherent if and only if $\MC(\varphi)$ is first-order definable.
\end{corollary}

\section{Outlook}
In this paper, we carried out a systematic investigation of the model checking problem for disjunctions of two dependence atoms. The work reported here suggests several different directions for future research, including the following:
\begin{enumerate}
\item Is it possible to classify the computational complexity of the model checking problem for all quantifier-free formulas of dependence logic $\DL$?
In particular, is there a dichotomy theorem between $\NP$-completeness and polynomial-time solvability for this problem?  Note that it is conceivable that every problem in $\NP$ is polynomial-time equivalent to the model checking problem of a quantifier-free $\DL$-formula. In that case, such a dichotomy theorem would be impossible (unless $\Ptime = \NP$), since by Ladner's theorem \cite{Ladner75}, if $\Ptime \not = \NP$, then there are problems in $\NP$ that are neither $\NP$-complete nor are in $\Ptime$.

\item By Proposition~\ref{coherent-FO-reduction}, if $\varphi$ is a coherent, quantifier-free $\DL$-formula, then $\varphi$ is equivalent to an $\FO$-formula $\varphi^*(T)$ that includes a relation symbol $T$ for the team. We conjecture that the converse is true, which would imply that coherence  coincides with  first-order definability for quantifier-free $\DL$-formulas. Our results confirm this conjecture for the case in which $\varphi$ is a disjunction of two dependence atoms.\footnote{In the meantime, this conjecture has been settled affirmatively in \cite{DBLP:journals/corr/BFHKKMV}.}

\item Study the \emph{implication problem} for disjunctions of two dependence atoms, i.e., given a finite set $\Sigma$ of disjunctions of two dependence atoms and a disjunction $\psi$ of two dependence atoms, does $\Sigma$ logically imply $\psi$?  
For functional dependencies (i.e., single dependence atoms), the implication problem is  solvable in polynomial time; moreover, there is a set of simple axioms, known as Armstrong's axioms~\cite[p.~186]{DBLP:books/aw/AbiteboulHV95}, for reasoning about functional dependencies. 
As argued in the introduction, disjunctions of two dependence atoms form a natural class of database dependencies that have not been studied in their own right thus far. 
Investigating the implication problem for disjunctions of two dependence atoms will enhance the interaction between dependence logic and database theory.

\end{enumerate}

\bibliography{bib}
\end{document}